%% file: main.tex
\documentclass[11pt]{article}
\usepackage[T1]{fontenc}
\usepackage[top=1in, bottom=1in, left=1in, right=1in]{geometry}

\usepackage[english]{babel}
\usepackage{bbold}
\usepackage{tikz}
\usepackage[colorlinks=true, allcolors=blue]{hyperref}
\usepackage[linesnumbered, ruled, vlined]{algorithm2e}
\usepackage{comment}
\usepackage{subcaption}
\usepackage{thm-restate}
\usepackage{amsthm,amsmath,amssymb}
\usepackage{cleveref}
\usepackage{graphicx}
\usepackage{enumitem}
\usepackage{xcolor}
\usepackage{bm}
\usepackage{natbib}
\usetikzlibrary{decorations.pathmorphing}
\usetikzlibrary{decorations.pathreplacing}
\usetikzlibrary{matrix,shapes,arrows,positioning,chains, calc}

\RequirePackage[type1,tt=false]{libertine}
\newcommand{\myangle}[1]{\left\langle #1\right\rangle}

\allowdisplaybreaks

\title{Online MMS Allocation for Chores}
\author{
Jiaxin Song$^{2}$ \\
\texttt{jiaxins8@illinois.edu}
\and
Biaoshuai Tao$^{1}$ \\
\texttt{bstao@sjtu.edu.cn}
\and
Wenqian Wang$^{1}$ \\
\texttt{wangwenqian@sjtu.edu.cn}
\and
Yuhao Zhang$^{1}$ \\
\texttt{zhang\_yuhao@sjtu.edu.cn}
\\[1ex]
\begin{tabular}{cc}
$^1$Shanghai Jiao Tong University &
$^2$University of Illinois Urbana-Champaign
\end{tabular}
}

\date{}

\crefname{algocf}{Algorithm}{Algorithms}

\input{macro}

\begin{document}
\maketitle

\begin{abstract}
We study the problem of fair division of indivisible chores among $n$ agents in an online setting, where items arrive sequentially and must be allocated irrevocably upon arrival. The goal is to produce an $\alpha$-MMS allocation at the end. Several recent works have investigated this model, but have only succeeded in obtaining non-trivial algorithms under restrictive assumptions, such as the two-agent bi-valued special case (Wang and Wei, 2025), or by assuming knowledge of the total disutility of each agent (Zhou, Bai, and Wu, 2023). For the general case, the trivial $n$-MMS guarantee remains the best known, while the strongest lower bound is still only $2$.

We close this gap on the negative side by proving that for any fixed $n$ and $\varepsilon$, no algorithm can guarantee an $(n - \varepsilon)$-MMS allocation. Notably, this lower bound holds precisely for every $n$, without hiding constants in big-$O$ notation, thereby exactly matching the trivial upper bound.

Despite this strong impossibility result, we also present positive results. We provide an online algorithm that applies in the general case, guaranteeing a $\min\{n, O(k), O(\log D)\}$-MMS allocation, where $k$ is the maximum number of distinct disutilities across all agents and $D$ is the maximum ratio between the largest and smallest disutilities for any agent. This bound is reasonable across a broad range of scenarios and, for example, implies that we can achieve an $O(1)$-MMS allocation whenever $k$ is constant. Moreover, to optimize the constant in the important personalized bi-valued case, we show that if each agent has at most two distinct disutilities, our algorithm guarantees a $(2 + \sqrt{3}) \approx 3.7$-MMS allocation.

\end{abstract}

\newpage
\input{10-intro}
\input{20-prelim}

\input{30-algo} 
\input{50-negative-warm}
\input{51-negative}
\input{60-2value}

\bibliographystyle{plainnat}
\bibliography{reference}

\end{document}

%% file: macro.tex
\newcommand{\A}{\mathcal{A}}

\newcommand{\MMS}{{\rm MMS}\xspace}
\newcommand{\game}{{\sc Stacking game}\xspace}
\newcommand{\abs}[1]{\left\vert#1\right\vert}

\newcommand{\x}{\phi}
\newcommand{\bc}{\bm{c}}
\newcommand{\bd}{\bm{d}}
\renewcommand{\A}{\mathcal{A}}
\newcommand{\B}{\mathcal{B}}
\newcommand{\X}{\mathcal{X}}

\newcommand{\tA}{\tilde{A}}
\newcommand{\tB}{\tilde{B}}

\renewcommand{\u}{\theta}

\newcommand{\jx}[1]{{\color{violet}[Jiaxin: #1]}}
\newcommand{\TBD}[1]{{\color{blue}(TBD: #1)}}
\definecolor{biaoshuaiFavourate}{rgb}{0.50, 0.30, 0.07}

\newcommand{\wq}[1]{{\color{blue}[wwq: #1]}}

\theoremstyle{plain}
\newtheorem{theorem}{Theorem}
\newtheorem{lemma}{Lemma}

\newtheorem{proposition}{Proposition}

\theoremstyle{definition} 
\newtheorem{definition}{Definition}

\newtheorem{example}[theorem]{Example}

\Crefname{equation}{Equation}{Equations}

%% file: 10-intro.tex
\section{Introduction}

We study the problem of fair division for indivisible chores (items with negative utilities). We have $n$ agents and $m$ chores, and our goal is to assign the chores to the agents. Each chore $j$ has a personalized disutility $d_i(j)$ for each agent $i$. The goal is to make every agent feel fair. 
There are various notions of capturing fairness.
These fairness notions can be divided into two categories: \emph{envy-based} and \emph{share-based} notions.
Notions falling into the former categories are comparison-based, with the typical example of \emph{envy-freeness}~\citep{gamow1958puzzle,Foley67,varian1973equity}, which requires that every agent weakly prefers her own allocated share to any other agent's.
This is not always possible if indivisible items are considered (e.g., $m<n$), and many relaxations such as \emph{envy-freeness up to one/any item} and their approximation variants are then studied~\citep{lipton2004approximately,budish2011combinatorial,caragiannis2019unreasonable,plaut2020almost,chaudhury2020efx,berger2022almost,amanatidis2024pushing}.
Share-based notions, on the other hand, typically set a threshold for each agent, and they require that each agent's utility is above this threshold.
Typical share-based notions include \emph{proportionality}~\citep{Steinhaus48}, which requires that each agent's utility (disutility, resp.) for her received allocation is at least (at most, resp.) her valuation of the entire item set divided by $n$ for allocating goods (chores, resp.).
That is, each agent's satisfaction must at least meet the average.
Same as envy-freeness, a proportional allocation is not guaranteed to exist when items are indivisible.
Two different types of relaxation have been studied in the past.
One type focuses on ``almost proportionality'': proportionality can be guaranteed by adding/removing one item (or a small number of items) to/from an agent's allocated bundle~\citep{DBLP:conf/sigecom/ConitzerF017,DBLP:conf/aaai/BarmanK19,DBLP:journals/orl/AzizMS20}.
The other type focuses on relaxing the ``average threshold'': instead of setting the threshold to be the average value, we consider the thresholds depending on the ``most balanced'' $n$-partition of the item set.
Typical notions of this type include \emph{max-min share} (for goods allocation) and \emph{min-max share} (for chores allocation), and their approximation versions~\citep{DBLP:conf/bqgt/Budish10}.


Traditional studies of these problems assume prior knowledge of all agents and items, and focus on discussing the existence of fair allocations or how efficiently such allocations can be constructed. This paper shifts the focus to the \emph{online} scenario. In many real-world settings, division processes occur over time in an online manner. For instance, when assigning jobs to employees on an online service platform, jobs typically arrive over time, and we must make irrevocable assignments at their arrival. This motivates the online fair division problem. We assume that the jobs arrive in a sequence, and we must irrevocably allocate each job to an agent according to this arrival order, and we aim to achieve fairness at the end. Note that the algorithm does not know any information about the future, so the only thing it can do is to guarantee fairness at every time. 
There are already several studies that focus on maintaining fairness in online settings~\citep{benade2018make,he2019achieving,zeng2020fairness,aleksandrov2020online,zhou2023multi,wang2024improved,neoh2025online,kulkarni2025online,amanatidis2025online,DBLP:journals/corr/abs-2505-24321}.

In our paper, we adopt the share-based notion of fairness \emph{the min-max share} (\MMS). 
Informally, a threshold $\MMS_i$ is defined for each agent $i$, which is the value/disutility of agent $i$'s least preferred bundle in the most balanced $n$-partition of the item set, and \MMS requires each agent's disutility for her allocated bundle is at most this threshold $\MMS_i$.
It is known that \MMS is not always achievable even in the offline version, so its approximation version has been studied in the past: an allocation is $\alpha$-approximate \MMS if each agent $i$'s disutility for her allocated bundle is at most $\alpha\cdot\MMS_i$.
In the offline setting, the best upper bound known is $15/13$~\citep{DBLP:conf/aaai/GargHS25}, while the best lower bound known is $44/43$~\citep{DBLP:conf/wine/FeigeST21}.

Our goal is to design an algorithm that maintains an approximate \MMS allocation throughout the \emph{online} arrival of items. We say that an algorithm achieves a competitive ratio of $\alpha$ of MMS if it guarantees an $\alpha$-MMS allocation in the online setting.
In \citet{zhou2023multi}, it is already shown that it is impossible to maintain an $\alpha$-MMS allocation for goods for any positive $\alpha$, even in a very special case where $n =2$. For chores, it is easy to see that any algorithm can trivially maintain an $n$-MMS allocation: just allocate all items to a single agent. In \citet{zhou2023multi}, it is shown that no algorithm can maintain an $\alpha$-MMS allocation for $\alpha < 2$. However, there is a lack of positive results in this setting. \citet{zhou2023multi} achieves an $O(1)$-MMS allocation only under a strong assumption that the total disutility (i.e., $\sum_{j=1}^m d_i(j)$) of each agent is known to the algorithm in advance. The result in \citet{DBLP:journals/corr/abs-2505-24321} achieves a $5/3$-MMS allocation online, but only in a very limited case where $n=2$ and each agent has only two types of values. In the general case, there remains a significant gap between $2$ and $n$. In this paper, we focus on the following main open question in this line of research.

\begin{center}
\emph{Can we improve upon the trivial bound of $n$ in the online MMS allocation problem for chores?}
\end{center}


\subsection{Our Contribution}

First, we address this main open question on the negative side. We construct a hard instance showing that for any fixed $n$ and $\varepsilon$, no algorithm can maintain an $(n - \varepsilon)$-MMS allocation in this problem; see \Cref{thm:lowerbound}. The proof is presented in \Cref{sec:negative}. This complements the trivial upper bound of $n$ and thus completely resolves the open question. 

However, we observe that this hard instance is arguably impractical for two main reasons.
\begin{itemize}
    \item The disutilities of different items can vary extremely widely. For example, we repeatedly construct new items whose disutility is roughly $1/\varepsilon$ times the total disutility of all previously arrived items in the instance.
    \item The number of distinct disutility values for each agent is extremely large. In the hard instance, the competitive ratio of MMS gradually increases to $n$ as more items arrive. Moreover, to achieve this, we need to introduce new disutility values repeatedly. As described earlier, each such value is roughly $1/\varepsilon$ times the total disutility of all previously arrived items.
\end{itemize}

Thus, while our hardness result closes the theoretical gap, it may not pose a significant practical concern. This leaves open the possibility of designing algorithms that maintain fair allocations in more realistic scenarios. In particular, can we achieve a small competitive ratio of MMS when each agent has only a limited number of distinct disutility values? For example, in the personalized bi-value case (where each agent has only two types of disutilities), is it possible to extend the $5/3$-MMS allocation from \cite{DBLP:journals/corr/abs-2505-24321} beyond the case of $n=2$? Moreover, let $d_i^{\max}$ and $d_i^{\min}$ denote the maximum and minimum disutilities of agent~$i$, respectively, and let $D$ be the maximum ratio $d_i^{\max} / d_i^{\min}$ across all agents. Can we achieve a good competitive ratio of MMS when $D$ is small? This is analogous to the familiar $p_{\max} / p_{\min}$ ratio in scheduling problems.

Our positive result is as follows. We propose an algorithm, which can be viewed as a generalized round-robin approach, that maintains a $\min\{n, O(k), O(\log D)\}$-MMS allocation in the online setting, where $k$ is the maximum number of value types (i.e., distinct disutilities) across all agents, and $D$ is the maximum ratio $d_i^{\max} / d_i^{\min}$; see \Cref{thm:Kvalued}. It means that, for every personalized $k$-value instance, assuming $k$ is a constant, our algorithm can maintain an $O(1)$-MMS allocation, regardless of $n$. For the special case of the personalized bi-value setting, we further optimize the constant and achieve an $2 + \sqrt{3} \approx 3.7$-MMS allocation, as shown in \Cref{thm:bi-value}.

Notably, our general algorithm does not require any prior knowledge of parameters related to the online arrival of items, such as $k$ or $D$. However, to achieve a better constant in the personalized bi-value case, the algorithm must be informed in advance that each agent has at most two types of disutility, for simplicity. We remark that this requirement of explicitly announcing the bi-value assumption can be removed by initially proceeding under the bi-value assumption and then reverting to the general algorithm if it is later observed that the number of distinct disutilities exceeds two. 

\subsection{Our Techniques}
Our algorithm design stems from the basic round-robin algorithm, which guarantees a $1$-MMS solution when all items have only one disutility per agent. It generalizes the core idea of round robin—ensuring that each agent receives at most one item in every $n$ rounds, to more general settings. In particular, we aim to preserve this ideal property for each agent with respect to each distinct type of item. Although it is difficult to achieve this exactly in the general case, we can enforce it approximately. To this end, we introduce a parameter called \emph{pressure}, denoted by $H_i^\u$, which quantifies whether agent $i$ has received too many type-$\u$ items relative to the ideal benchmark of receiving $1$ out of every $n$. We then apply a greedy allocation rule: each new item is assigned to the agent with the lowest pressure for its type.

\paragraph{Stacking Game: A Special Discrepency Minimization Problem.} Interestingly, the analysis of this process can be framed as a special discrepancy minimization problem. Concretely, we maintain an $nk$-dimensional vector initialized to zero. In each round, an adversary selects $n$ dimensions, and the algorithm responds greedily by increasing the smallest of these coordinates and decreasing the others, ensuring that the total change sums to zero and that the magnitude of change in each step is a constant. The central question is whether the algorithm can keep all coordinates small in the worst case across all adversarial choices.
To address this, we formulate an intriguing continuous and symmetric version of the problem, which we call \game, that captures all possible choices of $n$. This game may be of independent interest. We show that in this setting, the maximum coordinate value achieved remains $O(k)$. Translating this result back to the original MMS allocation problem, we conclude an $O(k)$ competitive ratio for MMS, combined with additional techniques such as value rounding.

\paragraph{Negative Result.}
First, we introduce a relatively simple idea to show that no algorithm can achieve a better-than-$2$-MMS guarantee in the case of $n=2$. The adversary begins by releasing an item with disutility $1$ for both agents. Without loss of generality, suppose agent~1 takes this item. The adversary then repeatedly releases items whose disutility is $\epsilon$ for agent~1 and is equal to $\frac{1}{\epsilon}$ times the total disutility of agent~2 over all previous rounds. By choosing $\epsilon$ to be a sufficiently small constant, there must eventually be a round in which agent~2 takes an item; otherwise, agent~1's allocation would approach $2$-MMS. At this point, the adversary releases a new item with disutility for agent~1 equal to $1$ plus the sum of all small items already taken by agent~1, and with disutility for agent~2 set equal to the disutility (s)he just took in the last round. As a result, whether agent~1 or agent~2 takes this item, their allocation is forced to be close to $2$-MMS.

There are two key ideas underlying this construction: 
(1) We create a simple disutility sequence for agent~1, which prevents agent~1 from taking all items indefinitely. 
(2) We generate a sharply growing sequence of disutilities for agent~2, where each new item effectively \emph{cleans up} all previously released items by making their disutilities negligible. However, these two simple ideas do not easily extend to the general case of $n$ agents for proving an $n$-MMS lower bound. We highlight some reasons for this difficulty, even when $n=3$. First, if we simultaneously apply the simple identical $\epsilon$-size disutility sequence to two different agents, these two agents can keep taking such items in a round-robin way forever, without ever forcing the third agent to take any new item. Second, since our goal is to establish a lower bound for $n \ge 3$, the clean-up idea would reduce the current allocation back down to a $1$-MMS level, even if we have already accumulated a $2$-MMS allocation for one agent. This makes it challenging to push the allocation beyond even the threshold of $2$-MMS in the multi-agent setting.

To address these issues without resorting to discussing an unmanageable number of algorithms' actions for $n$ agents, we propose a recursively constructed hard instance. We begin by presenting an alternative approach to reproduce the $2$-MMS hardness for $n=2$ in a more extensible way. Next, to prove the $n$-MMS hardness, we define a function $T(n')$ to denote the number of rounds required such that, if all items are taken by a set of $n' \leq n$ agents, then at least one of these $n'\leq n$ agents reaches an allocation of $n$-MMS, measured against the MMS benchmark for $n$ agents. For example, we have $T(1) = n$ if we repeatedly release items with disutility $1$. We then recursively construct an instance for $T(n'+1)$ using $T(n')$, effectively treating $T(n')$ as a black box. Specifically, we repeatedly apply $T(n')$ on agents~$1$ through $n'$, thereby forcing agent~$n'+1$ to take new items. Note that we also need to design a way to \emph{clean up} all agents $1$ to $n'$ before each invocation of $T(n')$. Finally, we carefully design the disutility sequence for agent~$n'+1$ to ensure not only that the sum of the disutilities of the missed items remains negligible, but also that the total disutility of the items taken becomes large enough to approach an $n$-MMS allocation, all within finitely many invocations of $T(n')$.

\subsection{Other Related Work}
The online MMS allocation problem for indivisible chores is closely related to the online load balancing problem. The online load balancing problem can be viewed as a special case of the online MMS allocation problem on chores, where all agents have identical disutility valuations for each chore. The study of online load balancing dates back to the 1960s~\citep{DBLP:journals/siamam/Graham69}. A substantial body of work has focused on determining the optimal competitive ratio for this clean setting~\citep{DBLP:journals/siamcomp/Albers99,DBLP:journals/jcss/BartalFKV95,DBLP:journals/jal/KargerPT96,DBLP:journals/ipl/BartalKR94,DBLP:journals/actaC/FaigleKT89,DBLP:conf/soda/GormleyRTW00,DBLP:conf/esa/FleischerW00,DBLP:journals/siamcomp/RudinC03}. Currently, the best known bounds place the competitive ratio in the range $[1.88, 1.92]$, where the lower bound is due to~\cite{DBLP:conf/esa/FleischerW00} and the upper bound is established by~\cite{DBLP:journals/siamcomp/RudinC03}.

Several works have also focused on maintaining an indivisible fair allocation in online settings, but from different perspectives. Some studies, like ours, examine the item arrival model, but concentrate on envy-based fairness notions. Most of these assume that the items are goods; we list them as follows. \citet{benade2018make} studies the notion of envy, aiming to minimize the total envy over a time horizon. \citet{he2019achieving} seeks to minimize the number of reallocations needed to maintain an EF1 allocation online. \citet{zeng2020fairness} and \citet{DBLP:journals/ior/BenadeKPPZ24} explore the trade-off between envy-based fairness and efficiency, with respect to Pareto-optimality. \citet{DBLP:conf/nips/ProcacciaS024} and \citet{DBLP:conf/aistats/YamadaKAI24} study the problem in a bandit setting. On the other hand, \citet{kulkarni2025online} also investigates MMS allocations but in a different online model: instead of items arriving online, they provide a comprehensive study of the case where agents arrive online. For broader overviews of results on fair allocation of items (such as divisible items), we refer readers to the surveys by \citet{DBLP:conf/aaai/AleksandrovW20} and \citet{DBLP:journals/ai/AmanatidisABFLMVW23}.

MMS allocation for chores is also a significant topic in the offline setting. \citet{DBLP:conf/aaai/AzizRSW17} first observes that exact MMS allocations do not always exist, motivating the study of approximate MMS allocations. Subsequent works have aimed to tighten these bounds, either by improving the best known approximation ratios \citep{DBLP:journals/teco/BarmanK20,DBLP:conf/sigecom/HuangL21,DBLP:conf/aaai/GargHS25,DBLP:conf/wine/FeigeST21} or by strengthening lower bounds on the achievable ratio. Currently, the best upper bound is $15/13$~\citep{DBLP:conf/aaai/GargHS25}, while the best lower bound is $44/43$~\citep{DBLP:conf/wine/FeigeST21}.

%% file: 20-prelim.tex
\section{Preliminaries}

\subsection{Model}
We study the problem of fair division of indivisible chores.
There is a set of $n$ agents and $m$ items (or chores).
Denote by $N=[n]$ and $M=[m]$ the set of agents and items, respectively. 
An \emph{allocation} $\A = (A_1, \ldots, A_n)$ is a partition of the item set, where $A_i$ is the bundle of items allocated to agent $i$.
Let $\Pi$ be the entire space of all allocations.
Each agent $i$ is associated with a function $d_i:\{0,1\}^{[m]}\rightarrow \mathbb{R}_{\ge 0}$, which captures the disutility that agent $i$ incurs from receiving chores. For simplicity, we write $d_i(\{j\})$ as $d_i(j)$.
Assume the disutility functions are \emph{additive}: $d_i(\{S\}\cup \{j\}) = d_i(S) + d_i(j)$ for any $j\notin S$.

A disutility function is said to be \emph{personalized $k$-valued} if there are at most $k$ possible values of $d_i(j)$ for every $i \in N$ among different $j \in M$. Let $V_i$ be the set of possible values of $d_i(j)$. Note that $V_i$ is allowed to be different for different agents.

Denote by $V_i^\u$ the $\u$-th value in the set $V_i$ and $M_i^\u$ be the set of items that agent $i$ has value $V_i^\u$, i.e., $M_i^\u = \{j\in M: d_i(j) = V_i^\u\}$.
We also refer to an item in $M_i^\u$ as \emph{an item of type-$\u$ for agent~$i$}.
Let $N_i^\u = |M_i^\u|$ be the number of type-$\u$ item for agent $i$. 

We consider an online setting where items arrive \emph{online} in an arbitrary order: in each iteration, an item arrives, and the allocation algorithm immediately assigns it to an agent. For convenience, we index the items in $[m]$ according to their arrival order. In particular, we use item~$j$ to refer to the $j$-th arriving item, and write $[j]$ for the set of the first $j$ arriving items. Note that any allocation made is irrevocable and cannot be changed in subsequent rounds. Below, we introduce the fairness notion considered in this paper.
%
\subsection{Fairness Notions}
\begin{restatable}[Minmax share (MMS)]{definition}{DefMMS}
The \emph{min-max share (MMS)} of agent $i$ is defined as 
$$
\MMS_i = \min_{\A\in \Pi}\max_{i'\in [n]} d_i(A_{i'})\,.
$$
We say that an allocation $\A$ is an MMS allocation for agent~$i$ if $d_i(A_i) \le \MMS_i$ holds. Moreover, an allocation is said to achieve $\alpha$-MMS for $\alpha \ge 1$ if every agent~$i$ receives a disutility of at most $\alpha \cdot \MMS_i$, that is, $\forall i \in [n]$$d_i(A_i) \le \alpha \cdot \MMS_i$. We also say that an algorithm achieves a competitive ratio of $\alpha$ with respect to MMS if it guarantees an $\alpha$-MMS allocation.
\end{restatable}


\begin{restatable}[MMS for type-$\u$ items]{definition}{DefiMMStype} 
We define the min-max share of items of type $\u$ for agent $i$, denoted by $\MMS_i^\u$, as the MMS value computed solely over the type-$\u$ items for agent $i$. Formally, let $\Pi(M_i^\u)$ denote the set of all possible allocations of the items $M_i^\u$. Then
\[
\MMS_i^\u = \min_{\A \in \Pi(M_i^\u)} \max_{i' \in [n]} d_i(A_{i'})~.
\]
In fact, in this setting we have the closed-form expression $\MMS_i^\u = \lceil N_i^\u / n \rceil$.
\end{restatable}

{
\begin{lemma}
\label{lem:mms_decomp}
For any agent $i\in [n]$, it holds that $\sum_{\u=1}^k (\MMS_i^\u - V_i^\u) \le \MMS_i \le \sum_{\u=1}^k \MMS_i^\u$, where $k$ is the number of distinct disutility of agent $i$. 
\end{lemma}
\begin{proof}
We first prove the right inequality.
For any $\u\in [k]$, let $\A^\u$ be the MMS allocation of type-$\u$ items for agent $i$.
Then we can construct a complete allocation by taking the union of the $k$ allocations: $\A= \bigcup_{\u \in [k]} \A^\u$, and $A_{i'}$ is the bundle of $i'$ in $\A$. 
It follows that $\forall i' \in [n]$, $d_{i'}(A_{i'}) \leq \sum_{\u=1}^k \MMS_i^\u$. Therefore $\MMS_i \le \MMS_i^\u$.

Next, we prove the left inequality. 
Since all the items of type $\u$ have the same disutility of $V_i^\u$, we have that the largest bundle in $\MMS_i^\u$ exactly have $\lceil N_i^\u / n\rceil \leq N_i^\u/ n + 1$ number of type-$\u$ items.
%
Therefore, we can obtain that 
$$
\MMS_i^\u = \lceil N_i^\u / n\rceil \cdot V_i^{\u} \le \frac{N_i^\u\cdot V_i^\u}{n} + V_i^\u = \frac1{n}\cdot d_i(M_i^\u) +  V_i^\u\,.
$$ 
By summing them up for all $\u \in [k]$, it follows that 
$$
\sum_{\u=1}^k \MMS_i^\u - V_i^\u \le \sum_{\u=1}^k \frac{1}{n}\cdot d_i(M_i^\u) = \frac1n\cdot d_i(M) \le \MMS_i,
$$
which completes the other side.
\end{proof}

%% file: 30-algo.tex
\section{Our Algorithm}

This section presents our main result, stated in the following theorem.

\begin{restatable}{theorem}{thmKvalued}
\label{thm:Kvalued}
\Cref{alg:kvalue} guarantees a $\min\{n, O(k), O(\log D)\}$-MMS allocation for the online fair division problem under indivisible chores, where $D = \max_{i\in [n]} \dfrac{\max_{j\in [m]} d_i(j)}{\min_{j \in [m]} d_i(j) } $ and $k$ denotes the maximum number of distinct disutility values across all agents.
\end{restatable}

\begin{algorithm}[tb]
\caption{Greedy over Pressure}
\label{alg:kvalue}
\SetAlgoLined
\KwIn{$n$ agents and $m$ indivisible chores arriving online.
Each agent $i\in [n]$ has a disutility function $d_i(\cdot)$.
The value of $d_i(j)$ can only be known when item $j$ arrives\;}
\KwOut{an allocation $\A=(A_1, \dots, A_n)$\;}

Initialize the bundle $A_i$ as empty set, and $k_i \gets 0$ for each agent $i$\;


\For{each arriving item $j=1,\cdots, m$}{
    \For{each agent $i \in [n]$}{
    Round up $d_i(j)$ to the nearest power of two \;
    \If{$d_i(j)$ is a new disutility value of $i$} {
        $k_i \gets k_i + 1$, $\u_i(j) \gets k_i$, $H_i^{k_i} \gets 0$ \;
        We call the item with value $d_i(j)$ as the type-$k_i$ item of $i$ \; 
    }
    \Else {
        $\u_i(j) \gets$ the type of item $j$ of agent $i$ \;
    }
    }
    Let agent $i \gets \arg\min_{i \in [n]} H_i^{\u_i(j)}$ \;
    Allocate item $j$ to agent $i$: $A_{i}\leftarrow A_{\iota}\cup \{j\}$\label{line:choose_argmin} \;
    Update $H_{i}^{\u_{i}(j)} \gets H_{i}^{\u_{i}(j)} + 1$ \;
    \For{each agent $i' \in [n] \setminus \{{i}\}$}{
        Update $H_{i'}^{\u_{i'}(j)} \gets H_{i'}^{\u_{i'}(j)} - \frac{1}{n-1}$ \;
    }
}
\Return{the final allocation $\A$}\;
\end{algorithm}

\subsection{The Algorithm and an Overview of the Analysis}
\paragraph{The Algorithm.}
First, we provide an overview of the algorithm. Our algorithm can be seen as an extension of the round-robin algorithm, which guarantees a $1$-MMS allocation in the personalized one-valued special case where each agent has only \emph{one type of item}. In the standard round-robin procedure, items are assigned to agents in a fixed, arbitrary order. To understand why this ensures a $1$-MMS allocation, consider a fixed agent $i$. Her MMS value increases after she receives the first item, then the $(n+1)$-th item, the $(2n+1)$-th item, and so on. In rounds $1$ to $n$, it is safe for her to receive one item, since her MMS benchmark is $1$ item; in rounds $n+1$ to $2n$, it is safe for her to receive two items, as her MMS benchmark becomes $2$ items; and so forth. Thus, the round-robin process naturally ensures that each agent receives no more than her MMS benchmark, thereby achieving a $1$-MMS allocation.

We can also view this process from another perspective by introducing a \emph{pressure} parameter for each agent $i$, denoted by $H_i$. When an item arrives and is allocated to agent~$i$, her pressure increases by $1$; if she does not receive the item, her pressure decreases by $\frac{1}{n-1}$. For example, if agent~$i$ receives the first item, her pressure becomes $1$, and if she does not receive any of the items from $2$ to $n$, her pressure gradually decreases back toward $0$. When she receives the $(n+1)$-th item, her pressure again becomes $1$. This pressure serves as a measure of the competitive ratio of MMS: as long as the pressure of each agent stays at most $1$, the allocation remains $1$-MMS. If an agent’s pressure reaches $2$, it indicates she has received one item more than her MMS allocation. By framing it this way, we can recover the round-robin algorithm through this notion of pressure: at each item’s arrival, we allocate it to the agent with the lowest pressure, and then update all agents’ pressures accordingly. It is easy to see that this maintains a $1$-MMS allocation, since the pressure of every agent remains at most $1$. Moreover, if we break ties using an arbitrary but fixed order, we precisely recover the classical round-robin procedure.

This new perspective allows us to extend the round-robin idea to the multi-value case. Assume that there are $k$ different value types for each agent; now we assign $k$ separate pressure variables to each agent~$i$, denoted by $H_i^\u$ for each $\u \in [k]$. When a new item $j$ arrives, it affects $n$ different pressures, each corresponding to the value type of a particular agent. We generalize the greedy approach from the single-value case by selecting the agent with the lowest pressure (with respect to the type of $j$ of the agent) to receive the new item. We then update the pressures as follows: if agent~$i$ receives item~$j$ and $j$ is of type~$\u$ for agent~$i$, we increase $H_i^\u$ by $1$. Conversely, if agent~$i$ does not receive item~$j$ and $j$ is of type~$\u$ for her, we decrease $H_i^\u$ by $\frac{1}{n-1}$. 

To complete the algorithm, we still need one additional component called \emph{value rounding}, a standard technique in algorithm design. For each arriving item $j$, we round up its disutility to the nearest power of $2$. This rounding step offers several advantages. First, since we treat all different value types separately, our competitive ratio of MMS could incur a multiplicative factor equal to the number of types when summing the competitive ratios across all types. However, by applying this rounding procedure and leveraging the properties of geometric sequences, we reduce this loss to a constant factor. In particular, with the rounding procedure, we can improve the ratio from $O(k^2)$ to $O(k)$. Second, value rounding naturally decreases the number of distinct types when the ratio between the maximum and minimum disutilities is small, which enables us to establish an $O(\log D)$ competitive ratio with respect to MMS, where $D$ denotes the maximum ratio $d_i^{\max} / d_i^{\min}$ over all agents $i$.

Finally, we refer to \Cref{alg:kvalue} for the detailed pseudocode of the complete algorithm.

\paragraph{The Analysis.}

To analyze whether the algorithm achieves a small competitive ratio of MMS, we first examine whether it is possible to keep every $H_i^\u$ small under the greedy approach in \Cref{alg:kvalue}. Note that difficulties arise when the number of distinct disutilities exceeds one (and not only under the greedy approach). Consider the following example: there are $2$ agents and $2$ types of disutility for each agent. The first item is of type $A$ for both agents, denoted by $(A,A)$, and we assume agent $1$ takes it. The second item is $(B,B)$, and we assume agent $2$ takes it. Then, a third item arrives, which is $(A,B)$. In this case, either agent $1$ ends up taking two consecutive items of type $A$, or agent $2$ takes two consecutive items of type $B$, due to the unavoidable mistakes inherent in online decision-making. Consequently, one of the $H_i^\u$ values grows to $2$. 
It is not hard to see that even if we had instead let agent $1$ take the second item, we could continue releasing items to force some $H_i^\u$ to approach $2$. It remains unclear whether this type of construction becomes worse as $n$ increases, or whether it can be upper bounded by a constant over all choices of $n$ under this special two-value case.

We observe that this corresponds to a special form of a discrepancy minimization problem. We can view the problem as maintaining an $nk$-dimensional vector initialized to zero. Note that we abuse the definition of $k$ to mean the number of types after rounding; it is at most the original value of $k$, which means the number of types before rounding, and also at most $O(\log D)$. At each time step, the adversary reveals $n$ special dimensions, and we must choose one of them to increase by $1$, while each of the remaining dimensions decreases by $\frac{1}{n-1}$. In response, our algorithm follows a greedy strategy: it always selects the dimension with the lowest current value to increase. We are interested in whether this approach can keep the values across all dimensions small. Specifically, our goal is to show that the maximum value attained depends only on $k$, so that for any fixed $k$, the values remain bounded regardless of whether $n$ is small, large, or even tends to infinity.

To simultaneously analyze all values of $n$, we abstract the problem into a continuous and symmetric formulation, which we call the \game. This formulation is analytically friendly and generalizes the discrete problem across all values of $n$. In this framework, we represent the $nk$-dimensional vector as a function plotted over the interval $(-0.5,0.5]$. The adversary's choice of $n$ dimensions then corresponds to selecting a union of intervals with total measure $1/k$. In response, our algorithm increases the lower part of the function and decreases the higher part according to a greedy rule. We will show that under this process, the adversary cannot force the values to grow arbitrarily large; instead, they remain upper bounded by $O(k)$.

Finally, we combine the contributions from all different value types to establish the overall competitive ratio of \MMS for our algorithm. This incurs only a constant multiplicative loss because of our rounding procedure, so the final competitive ratio remains $O(k)$.

\subsection{Stacking Game and the Competitive Ratio of MMS}

\begin{definition}[\game]
\label{def:stacking_game}
Let $I$ be the interval $I=(-1/2, 1/2]$ and $k\ge 1$ be a given number.
An adversary holds a non-decreasing function $f: I\rightarrow \mathbb{R}$ defined over the interval.
(S)he starts with $f = f^0$ with $f^0(x) = 0$ for any $x\in I$.
Next, at each round $t = 1, \dots$, the adversary is allowed to update the function $f^{t-1}$ by the following way:
\begin{enumerate}[itemsep=0pt]
    \item[(i)] (S)he specifies a four-tuple
    \[
    \x^t = (a^t, b^t, A^t, B^t),
    \]
    where $a^t$ and $b^t$ are positive rational numbers not exceeding $1$, and $A^t$ and $B^t$ are unions of finitely many mutually disjoint left-open, right-closed intervals (with rational endpoints) contained in $I$. These sets satisfy
    \[
    |A^t| = \frac{1}{k} \cdot \frac{b^t}{a^t + b^t}~, 
    \quad |B^t| = \frac{1}{k} \cdot \frac{a^t}{a^t + b^t} ~,
    \quad \text{so that} \quad |A^t| + |B^t| = \frac{1}{k}.
    \]
    Moreover, $A^t$ must lie entirely to the left of $B^t$. Formally, for any $x \in A^t$ and $y \in B^t$, we have $x \leq y$.
    \item[(ii)] Next, the value of every point in $A^t$ is increased by $a$, and the value of every point in $B^t$ is decreased by $b$ while the other values remain the same. 
    $$
        g^{t-1}(x) = 
        \begin{cases}
            f^{t-1}(x) + a & x \in A^t \\
            f^{t-1}(x) -b  & x \in B^t \\
            f^{t-1}(x) & \text{Otherwise}.
        \end{cases}
    $$
    \item[(iii)]
    The function values of $g^{t-1}$ are permuted into non-decreasing order to obtain the next round non-decreasing function $f^{t}$. We will formally define this operation later.
\end{enumerate}
The \emph{objective} of the adversary is to maximize the maximum of $f$, $\max_{x\in I}f(x)$.
The adversary can terminate the game at the end of any round, and the objective value at the termination is called the \emph{gain} of the adversary.
\end{definition}
We illustrate the stacking game through a simple example as follows.
\begin{example}[Running example for Stacking game]
Consider an example of $k=1$ as shown in \Cref{fig:stacking-example1}.
At each round, the adversary can choose two unions of subintervals with a total length of $1/k = 1$.
Fig.~\ref{fig:stacking-example1} illustrates the game through an example operation.
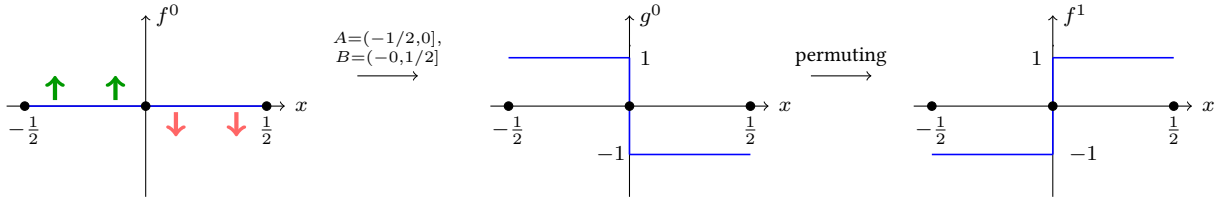
\begin{figure}[h]
\centering
\input{figure/stacking-example1}
\caption{An example operation of the stacking game with $k=1$.}
\label{fig:stacking-example1}
\end{figure}
Initially, the function $f$ is set as $f=f^0: x \mapsto 0$, and the maximum value of $f$ equals zero.
At the first round, the adversary sets $a^1=b^1=1$ and chooses two unions of intervals $A^1$ and $B^1$ as follows:
$$
A^1 = (-1/2, 0], \quad B^1 = (0, 1/2],
$$
each of which has the same length of $1/2$.
Next, the values within $A^1$ are lifted by $1$ while the values within $B^1$ are decreased by $1$, as shown in the middle figure.
Lastly, we permute the values within the entire interval $I$ and obtain a new function $f^1$, as shown in the right subfigure.
If the adversary terminates the game here, the objective value she gained from the game equals one.
\end{example}
\paragraph{How to permute the function formally?} In this paragraph, we formally define a sorting procedure by which the adversary permutes the function $g^{t-1}$ to obtain $f^{t-1}$. We prove that this procedure can always be carried out inductively. Moreover, we show that each function $f^t$ and $g^t$ remains piecewise constant throughout the process, with all breakpoints located at rational numbers.

For the base case, note that $f^0$ is piecewise constant with only one piece on $(-1/2,1/2]$. Then, assuming that $f^{t-1}$ is piecewise constant (with rational breakpoints), we show that $g^{t-1}$ is also piecewise constant with $O(t)$ pieces (with rational breakpoints). Consequently, we can sort these pieces to construct $f^t$, which by construction remains piecewise constant with $O(t)$ pieces (with rational breakpoints).

For ease of analysis, we introduce a bijective mapping that explicitly tracks how the points in $f^{t-1}$ are shifted to the points in $f^t$:
\[
x \mapsto \pi^t(x).
\]
Specifically, we proceed as follows. Since $f^{t-1}(x)$ and $g^{t-1}(x)$ are both piecewise constant, and since $f^t$ is obtained by sorting the pieces of $g^{t-1}$, we naturally obtain a bijection between these pieces from the sorting process. Let $L = (p, q]$ and $L' = (p', q']$ be corresponding pieces in $g^{t-1}$ and $f^t$, respectively. For each point $x = p + \delta \in L$, we define
\[
\pi^t(x) = p' + \delta.
\]

\paragraph{From Online-MMS to the Stacking Game.}
Then, we formally state the relation between the \game and the competitive ratio of MMS achieved by \cref{alg:kvalue}.
The main result is stated in \cref{lem:stackonlinemms}: once the gain of the adversary is upper bounded by an integer $\alpha$, \cref{alg:kvalue} will always output a $(4\alpha+2)$-MMS allocation.
In the next subsection, we provide an upper bound of the adversary's gain of $2k$ for any stacking game, which immediately implies \cref{thm:Kvalued}. 

\begin{restatable}{lemma}{lemStackOnlineMMS}
\label{lem:stackonlinemms}
Given $k$, if the adversary cannot gain more than $\alpha \in \mathbb{N}^+$ in the \game, then the output of Algorithm~\ref{alg:kvalue} is $(4\alpha+2)$-MMS, if the number of distinct disutility is at most $k$ for every agent. 
\end{restatable}
\begin{proof}
We aim to show that the execution of Algorithm~\ref{alg:kvalue} can be interpreted as an adversarial process of playing the \game. 
Consider a \game with the same parameter $k$ as in the online fair division instance. 
We divide the interval $(-1/2, 1/2]$ into $nk$ disjoint subintervals 
\[
(-1/2, -1/2 + 1/nk], \; \ldots, \; (1/2 - 1/nk, 1/2],
\]
and show that the $nk$ pressures $H_i^u$ can always be mapped to these $nk$ disjoint subintervals such that the value of $H_i^u$ equals the function value $f$ on the corresponding subinterval.

We prove this by induction on $t$. The base case holds trivially since all values are initially zero. 
Next, we analyze the change in $H$ and $f$ after one round. The adversary for Algorithm~\ref{alg:kvalue} releases a new item $j$, which corresponds to $n$ different $H_i^u$. Meanwhile, the adversary in the \game selects $n$ subintervals that match these values. Specifically, it sets $A^t$ to be the leftmost subinterval and $B^t$ to be the remaining $n-1$ subintervals, with parameters $a^t = 1$ and $b^t = \frac{1}{n-1}$.

As a result, the $n$ affected values $H_i^u$ can still be mapped to these $n$ subintervals after the update: the smallest one increases by $1$ and the other $n-1$ decrease by $\frac{1}{n-1}$. Thus, we maintain a valid mapping of the $nk$ pressures $H_i^u$ to the $nk$ disjoint subintervals throughout the execution.
%
%

Consider an example of $n=3$ and $k=2$, as shown in Figure~\ref{fig:algo-stacking}.
We label the $nk$ subintervals using $H_i^\u$.
Initially, all values are set to zero.
When the first item $1$ arrives, we compare the three values $H_1^1, H_2^1$, and $H_3^1$, at line~\ref{line:choose_argmin}.
As all three values are the same, without loss of generality, we allocate the item $1$ to agent 1.
Thereafter, we permute the function values and relabel the $nk$ intervals on the x-axis, as shown in the right figure.
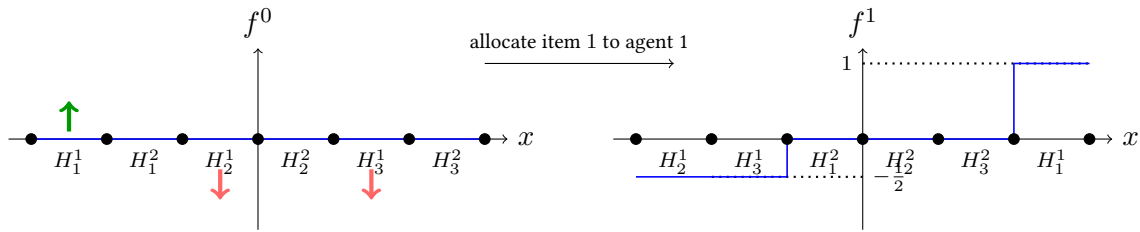
\begin{figure}[h]
\centering
\input{figure/algo-stacking}
\caption{An example of mapping from pressures to subintervals when $n=3$ and $k=2$.}
\label{fig:algo-stacking}
\end{figure}

Since the adversary cannot gain more $\alpha$ from the \game, we can obtain that no value of $H_i^u$ can exceed $\alpha$ throughout Algorithm~\ref{alg:kvalue} either.
Below, we prove that it then implies that the output allocation is $(4\alpha+2)$-MMS.
We first prove that the output allocation is $(2\alpha+1)$-MMS on the rounded instance.

\begin{proposition}
\label{prop:mms_of_rounded}
The output allocation of Algorithm~\ref{alg:kvalue} on the rounded instance is a $(2\alpha + 1)$-MMS allocation.
\end{proposition}

\begin{proof}[Proof of \Cref{prop:mms_of_rounded}]
We slightly abuse the notation and continue to use $M_i^\u, N_i^\u$, and $V_i^\u$ for the rounded instance, and we call an item of type-$\u$ of agent $i$ if its disutility is the $\u$-th largest in $V_i$. Since the instance is rounded, we have $V_i^1 \ge 2 V_i^2 \ge \dots \ge 2^{k-1} V_i^k$. Note that if $i$ has fewer than $k$ distinct values, we can pad with some dummy types whose disutilities follow the same geometric sequence, while ensuring that $H_i^\u = 0$ for these dummy types.

Let $A_i$ be the bundle $i$ gets in the output allocation. Observe that for every agent~$i$ and type~$\u$,
\begin{align*}
H_i^\u 
&= \bigl|\text{items in } M_i^\u \text{ allocated to } i\bigr| 
   - \frac{1}{n-1} \cdot \bigl|\text{items in } M_i^\u \text{ allocated to other agents}\bigr| \\
&= |A_i \cap M_i^\u| 
   - \frac{1}{n-1} \cdot \bigl(N_i^\u - |A_i \cap M_i^\u|\bigr) 
= \frac{n}{n-1} |A_i \cap M_i^\u| - \frac{1}{n-1} N_i^\u.
\end{align*}
Since $H_i^\u \le \alpha$, it follows that
\[
|A_i \cap M_i^\u| 
\le \frac{N_i^\u}{n} + \frac{n-1}{n} \cdot \alpha 
< \lceil N_i^\u / n \rceil + \alpha.
\]
Because $|A_i \cap M_i^\u|$ and $\alpha$ are integers, we can refine this to
\[
|A_i \cap M_i^\u| 
\le \lceil N_i^\u / n \rceil - 1 + \alpha.
\]
Since there must exist a bundle with at least $\lceil N_i^\u / n \rceil$ items in any allocation of $M_i^\u$, we obtain
\begin{align}\label{ineq:alg_ij_ub}
V_i^\u \cdot |A_i \cap M_i^\u| 
\le \lceil N_i^\u / n \rceil \cdot V_i^\u - V_i^\u + \alpha \cdot V_i^\u 
\le \MMS_i^\u - V_i^\u + \alpha \cdot V_i^\u.
\end{align}
Summing over all $\u \in [k]$ and applying \Cref{lem:mms_decomp}, we get
\begin{align*}
d_i(A_i) 
= \sum_{\u=1}^k V_i^\u \cdot |A_i \cap M_i^\u| 
&\le \sum_{\u=1}^k \bigl(\MMS_i^\u - V_i^\u\bigr) 
     + \alpha \sum_{\u=1}^k V_i^\u \tag{by~\Cref{ineq:alg_ij_ub}} \\
&\le \MMS_i + \alpha V_i^1 + \alpha \sum_{\u=2}^k V_i^\u \tag{by \Cref{lem:mms_decomp}}\\
&\le \MMS_i + \alpha V_i^1 + \alpha V_i^1 
= (2\alpha + 1) \MMS_i,
\end{align*}
where the last inequality uses the property $V_i^1 \ge 2 V_i^2 \ge \dots \ge 2^{k-1} V_i^k$, so the sum of the remaining values is bounded by $V_i^1$. This completes the proof.
\end{proof}

Finally, we show that the output allocation of Algorithm~\ref{alg:kvalue} on the original instance is $(4\alpha+2)$-MMS.
Since the disutility of every item is rounded up to the nearest power of two, i.e., $d_i(j) \le d_i'(j) < 2\cdot d_i(j)$, then the disutility of every agent in the original instance is at most equal to her disutility in the rounded instance, i.e., $d_i(A_i) \le d_i'(A_i)$.
Meanwhile, since every bundle's disutility is at least half of that in the rounded instance, then $\MMS_i \ge \MMS_i'/2$.
Therefore, it follows that $d_i(A_i) \le (4\alpha+2)\cdot \MMS_i$.
\end{proof}

In the next section, we will move to analyze the gain of the stacking game, and prove it is at most $2k$. 

\begin{restatable}{theorem}{thmStackGame}
\label{thm:ratio_of_stacking_game}
The adversary cannot gain more than $2k$ from the \game.
\end{restatable}
Assuming its correctness, we can now complete the proof of \Cref{thm:Kvalued}.

\begin{proof}[Proof of \Cref{thm:Kvalued}]
Here, the parameter $k$ in the \game actually refers to the number of distinct disutilities after rounding, which is at most the original $k$ and also bounded by $O(\log D)$. By applying \Cref{lem:stackonlinemms} and \Cref{thm:ratio_of_stacking_game}, our algorithm can guarantee an $(8k + 2)$-MMS allocation, with $k$ being the number of distinct disutilities after rounding. It implies the $O(k)$-MMS for the original $k$, and also $O(\log D)$-MMS. Combining with the fact that any allocation is $n$-MMS, we obtain the desired result, thus proving the theorem.
\end{proof}

\subsection{Analysis of The Stacking Game}
In this section, we aim to prove \Cref{thm:ratio_of_stacking_game}. At the beginning, we present a natural but important invariant for this process. 
\begin{proposition}
\label{inv:stacking_f} 
for all rounds $t$, 
$\int_{-1/2}^{1/2} f^t(x) dx = 0$.
\end{proposition}
\begin{proof}
    It is easy to see, at each step, that the total increment is always equal to the total decrement, so the lemma naturally holds. 
\end{proof}

We denote by $F^t(x,y)$ the integral of $f^t$ over the interval $[x, y]$, i.e., $F^t(x,y) = \int_x^{y} f^t(u) \, du$.
Let $F^t(x)$ be the integral of $f$ over the interval $[x, 1/2]$, i.e., $F^t(x) = F^t(x, 1/2) = \int_x^{1/2} f(u) \, du$.
We aim to transfer the original task to bounding $F^t(x)$. Below, we show that a good bound for $F^t(x)$ suffices to give a good bound for $f^t(x)$.  
\begin{lemma}\label{lem:bound_F}
For any fixed $t$, if $\forall x\in [-\frac{1}{2},\frac{1}{2}]$, $F^t(x)\leq k/2-2kx^2$, then $\forall x\in [-\frac{1}{2},\frac{1}{2}]$, $f^t(x)\leq 2k$ always holds.
\end{lemma}
\begin{proof}
Since $f^t(x)$ is a piecewise constant function on left-open, right-closed subintervals, 
and non-decreasing, we can obtain that
\begin{align*}
\max_{x\in I}f^t(x) & = f^t(1/2) = -\frac{d F^t(x)}{dx}\bigg|_{x= \frac12}  = - \lim_{\delta\rightarrow 0^+}\frac{F^t(1/2)- F^t(1/2- \delta)}{1/2 - (1/2 -\delta)} \\
& = \lim_{\delta\rightarrow 0^+}\frac{F^t(1/2- \delta)}{\delta} \tag{as $F^t(1/2) = 0$} \\
& \le \lim_{\delta\rightarrow 0^+}\frac{k/2- 2k(1/2 - \delta)^2}{\delta} = 2k~,
\end{align*}
which completes the proof.
\end{proof}

In the following proof, we show that the inequality $F^t(x) \le \frac{k}{2} - 2k x^2$ holds throughout the game. We prove this by induction on $t$. It is clear that the initial functions $f^0$ and $F^0$ satisfy this bound. Assuming as the induction hypothesis that $F^{t}(x) \le \frac{k}{2} - 2k x^2$, we aim to prove that the same inequality holds for $F^{t+1}$ as well.

The key question is whether there exists a ``bad'' choice of $\x$ that could turn a valid $F^{t}(x)$ into an invalid $F^{t+1}(x)$. To address this, we first narrow our discussion by showing that it suffices to consider choices of $A$ and $B$ where $A \cup B$ forms a single contiguous interval.

\begin{restatable}{lemma}{lemcontiguousInterval}
\label{lem:contiguous_interval}
At any round $t+1$, suppose the adversary applies the operation $\x = (a, b, A, B)$ on $f^{t}$, resulting in an updated function $f^{t+1}$ with integral $F^{t+1}$. Then there always exists another operation $\tilde{\x} = (a, b, \tA, \tB)$ such that $\tA \cup \tB$ is a single \textbf{contiguous} interval, and we have $\tilde{F}^{t+1}(y) \ge F^{t+1}(y)$ for all $y \in I$, where $\tilde{F}^{t+1}$ denotes the integral function of $\tilde{f}^{t+1}$ from $y$ to $1/2$, after applying $\tilde{\x}$.
\end{restatable}


Before presenting the proof of \Cref{lem:contiguous_interval}, we first give some intuition behind why the adversary can find such a contiguous operation.
Consider the following example:
\begin{example}
Let $k=2$ and $a=b=1$. At round $t+1$, suppose the function $f^t$ is defined by $f^t(x) = -1$ for $x \le 0$ and $f^t(x) = 1$ for $x > 0$. Consider the case where the adversary applies an operation with 
$A = (-\frac18, 0] \cup (\frac18, \frac14]$ and $B = (\frac14, \frac12]$, which is not contiguous.

We can then construct a new interval $\tilde{A}$ by merging the two subintervals of $A$ into a single contiguous interval: specifically, we remove $(-\frac18, 0]$ from $A$ and add $(0, \frac18]$ in its place. While this adjustment keeps the total integral unchanged, it shifts more of the mass to the right. As a result, the integral from any point $y$ to the rightmost endpoint $1/2$ is increased. Formally, we have $\tilde{F}^{t+1}(y) \ge F^{t+1}(y)$ for all $y \in I$. We show the process in Figure~\ref{fig:ioc}.

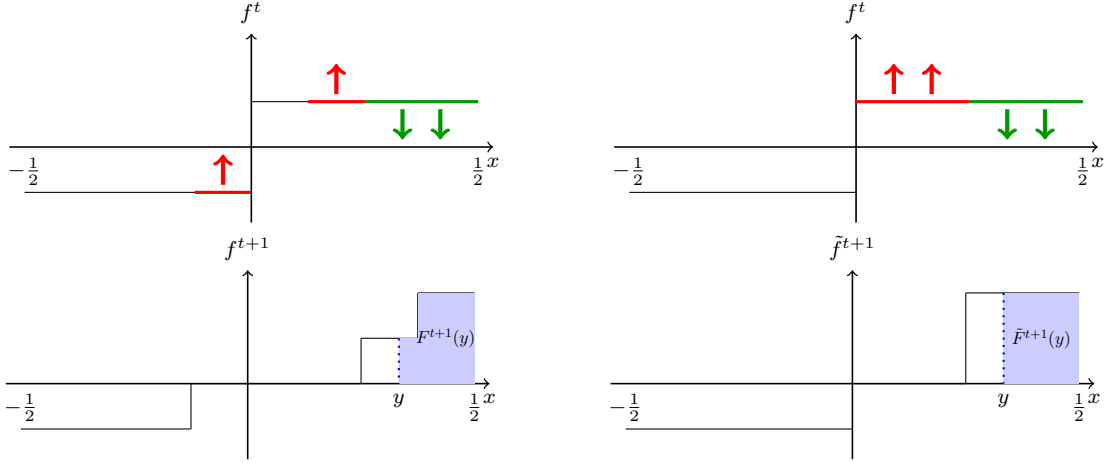
\begin{figure}[h]
\centering
\input{figure/idea_of_contiguous.tex}
\caption{An example showing $\tilde{F}^{t+1}(y) \ge F^{t+1}(y)$ with contiguous interval $\tilde{A}\cup\tilde{B}$.}
\label{fig:ioc}
\end{figure}
\end{example}
\begin{proof} [Proof of \Cref{lem:contiguous_interval}]
Denote by $e_1^A, \dots, e_s^A$ and $e_1^B, \dots, e_t^B$ the set of disjoint subintervals contained by $A$ and $B$ respectively.
Let $\ell_i^A$ (resp. $\ell_i^B$) and $r_i^A$ (resp. $r_i^B$) respectively denote the left and right endpoints of $e_i^A$ (resp. $e_i^B$).
By the definition of \game, we know that each of $\ell_i^A$, $r_i^A$, $\ell_i^B$, and $r_i^B$ is rational for every $i$. By the property that $f^{t}$ is piecewise constant, with all endpoints on rational numbers,  we can partition the small enough subinterval $I$ evenly into $\mathcal{I}=\{I_1,I_2,\ldots\}$, where each subinterval has length $\epsilon = \frac{1}{|\mathcal{I}|}$, such that: (1) $f^{t}$, $g^{t}$, and $f^{t+1}$ are all constant on each $I_i$, and (2) each subinterval $I_i$ either does not intersect $A$ and $B$ at all, or is entirely contained within $A$ or entirely within $B$.

Next, we perform the following update by piecing together subintervals of $A$ while maintaining $B$ unchanged.
If $A$ is not contiguous, then there exists some interval $I_A^*$ that does not lie in $A$, while parts of $A$ lie on both its left and right. Note that by the definition of the \game, $I_A^*$ cannot be in $B$. 
Denote by $I_A^{\min}$ the leftmost interval of $\mathcal{I}$ that is included by $A$.
Then we define $\tilde{A}$ by substituting $I_A^{\min}$ with $I_A^*$, which is still a valid operation.
Suppose the original values of $f^{t}$ over $I_A^{\min}$ and $I_A^*$ are respectively $f_A^{\min}$ and $f_A^*$.
We can see that the values of $g^t$ over $I_A^{\min}$ and $I_A^*$ are respectively $f_A^{\min} + a$, $f_A^*$ while the values of $\tilde{g}^t$ are respectively $f_A^{\min}$, $f_A^* + a$.
Denote by $G$ and $\tilde{G}$ the sets of function values on different intervals in $\mathcal{I}$, with respect to $g^t$ and $\tilde{g}^t$, respectively. 


Now we show that the value of $\tilde{F}^{t+1}(y)$ is weakly larger than $F^{t+1}(y)$ for any $y\in I$, where $F^{t+1}(y)$, and $\tilde{F}^{t+1}(y)$ are integrals of $f^{t+1}$, and $\tilde{f}^{t+1}$ from $y$ to $1/2$. 
Since both functions $f^{t+1}$ and $\tilde{f}^{t+1}$ are piecewise-constant on every $I_i$, $F^t$ and $\tilde{F}^t$ are piecewise-linear on every $I_i$. 
It suffices to prove that $\tilde{F}^{t+1}(y)$ is weakly higher than $F^{t+1}(y)$ at every endpoint of $\mathcal{I}$. Since all subintervals of $\mathcal{I}$ have equal length and $f^{t+1}$ and $\tilde{f}^{t+1}$ are obtained by sorting the values in $G$ and $\tilde{G}$ in ascending order, respectively, it suffices to show that for any $i \in [|\mathcal{I}|]$, the sum of the largest $i$ values of $\tilde{G}$ is weakly larger than that of $G$.
Formally, 
\begin{align}\label{eqn:top_i}
\mathrm{Top}_i(\tilde{G}) \ge 
\mathrm{Top}_i(G), \quad \text{for any } i\in [|\mathcal{I}|],
\end{align}
where $\mathrm{Top}_i(G)$ and $\mathrm{Top}_i(\tilde{G})$ are the sums of the largest $i$ values of $G$ and $\tilde{G}$ respectively.

The proof of \eqref{eqn:top_i} is as follows.
Let us denote by $g_i$ the rank-$i$ element in $G$, equivalently, the $i$-th largest element in $G$ (break tie arbitrarily).
Recall that from $G$ to $\tilde{G}$, only two elements change: namely, $f_A^{\min}+a$ and $f_A^*$,
which become $f_A^{\min}$ and $f_A^*+a$ in $\tilde{G}$, respectively. We now define an indexing rule for elements in $\tilde{G}$. 
For all elements except these two changed ones, we keep the same index as their rank in $G$, 
so that $\tilde{g}_i$ denotes the element originally ranked $i$ in $G$.
Let $i_1$ and $i_2$ be the ranks of $f_A^{\min}+a$ and $f_A^*$ in $G$, respectively, and without loss of generality, assume $i_1 < i_2$. We use $i_1$ as the index of $f_A^* + a$ and $i_2$ as the index of $f_A^{\min}$. That is, $\tilde{g}_{i_1} = f_A^*+a$ and $\tilde{g}_{i_2} = f_A^{\min}$. By construction, we observe that 
\[
    \tilde{g}_{i_1} \geq g_{i_1}, \quad \tilde{g}_{i_2} \leq g_{i_2},
    \quad 
    \tilde{g}_{i_1} - g_{i_1} = g_{i_2} - \tilde{g}_{i_2} \quad \text{and} \quad \tilde{g}_{i} = g_{i}, \text{ for other $i\in [\abs{\mathcal{I}}]$}~.
\]

Therefore, under this indexing, it follows directly that
\begin{equation}
    \label{eqn:top_i_1}
    \sum_{j=1}^{i} \tilde{g}_j \geq \mathrm{Top}_i(G), \quad \text{for any } i\in [|\mathcal{I}|]~.
\end{equation}

Next, we compare $\sum_{j=1}^{i} \tilde{g}_j$ with $\mathrm{Top}_i(\tilde{G})$.
By definition, $\mathrm{Top}_i(\tilde{G})$ is the maximum sum over any $i$ elements in $\tilde{G}$,
so we clearly have
\[
    \mathrm{Top}_i(\tilde{G}) \geq \sum_{j=1}^{i} \tilde{g}_j, \quad \text{for any } i\in [|\mathcal{I}|]~.
\]

Combining this with \eqref{eqn:top_i_1}, we complete the proof of \eqref{eqn:top_i}.

Lastly, by continuing to apply the above update rule at most $|\mathcal{I}|$ times, which is finite for any fixed $t$, we can make $A$ contiguous.
Similarly, when $B$ is not contiguous, then we find a subinterval $I_B^*\notin B$ with parts of $B$ lying on both its left and right. 
Then we find the rightmost subinterval $I_B^{\max}$ of $B$ and update $B$ by excluding $I_B^{\max}$ and including $I_B^*$.
Using a symmetry argument, we can continue applying it and eventually make $B$ contiguous as well, which also does not decrease the value of $F^{t+1}(y)$ for any $y\in I$.
Using the same argument, we can further combine intervals $A$ and $B$, by finding a subinterval $I^*$ not in $A$ and not in $B$ (in fact it is in the middle of $A$ and $B$), and constructing $\tilde{A}$ by substituting $I_A^{\min}$ with $I^*$. This construction can also guarantee $\tilde{F}^{t+1}(y)$ is weakly larger than $F^{t+1}(y)$ for any $y \in I$, by the same proof before. Combining all these cases, we can construct a single contiguous interval $\tilde{A} \cup \tilde{B}$ that satisfies the properties we need, in finite (totally at most $|\mathcal{I}|$) steps. This concludes the lemma.
\end{proof}

Next, we prove that there is no choice of $\phi$ that can turn a valid $F^t$ into an invalid $F^{t+1}$. By the lemma above, it suffices to consider only contiguous choices of $A$ and $B$.

\begin{restatable}{lemma}{lemUbofcontiguousInterval}
\label{lem:ub_of_contiguous_interval}
If $F^t(x)\le k/2 - 2kx^2$ holds, then for any operation $(a, b, A, B)$ with contiguous operation interval $A\cup B$, $F^{t+1}(x)\le k/2 - 2kx^2$ still holds.
\end{restatable}
\begin{proof}
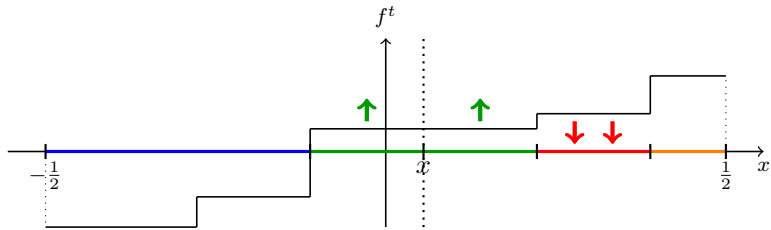
\begin{figure}[h]
        \centering
        \input{figure/induction-pic1}    
        \caption{An example of $(a, b, A, B)$ and $f^t$, where $x$ is to the left of $B$.}
        \label{fig:contiguous-operation}
\end{figure}
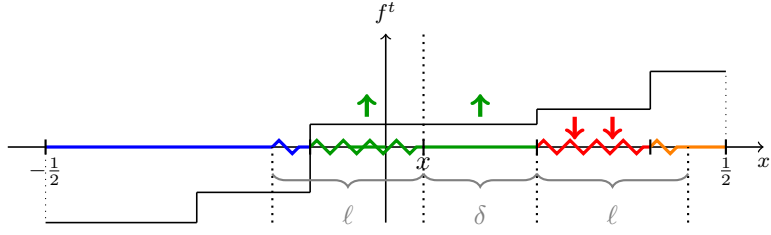
\begin{figure}[h]
        \centering
        \input{figure/induction-pic2}    
        \caption{An example for the upgrading interval and the downgrading interval.}
        \label{fig:ar_dl_bl}
\end{figure}
Given an operation $(a, b, A, B)$ that satisfies the precondition, we first prove the lemma for any point $x$ to the left of $B$
Here, we say that a point $x$ is to the left of $B$ if $x \leq y$ holds for every $y \in B$, as illustrated in Fig.~\ref{fig:contiguous-operation}. The remaining case, where $x$ is to the right of $A$, will be addressed at the end of the proof by a symmetric argument.


We define a point $y$ to be an \emph{upgrading point} if $y \le x$ and $\pi(y) > x$, and a \emph{downgrading point} if $y > x$ and $\pi(y) \le x$. The key structural property of these two sets of points is given below.
\begin{proposition}
The set of upgrading points and the set of downgrading points each form a single contiguous interval of the same length, denoted by $\ell$. Moreover, the set of upgrading points, called the \emph{upgrading interval}, must end at $x$ and is given by $(x-\ell, x]$. The set of downgrading points, called the \emph{downgrading interval}, must start at the left endpoint of $B$ and is given by $(x+\delta, x+\delta+\ell]$, where $x+\delta$ represents the left endpoint of $B$. (See the two zigzag lines in \Cref{fig:ar_dl_bl} as illustrations of these two intervals.)
\end{proposition}

\begin{proof}
First, consider the set of upgrading points. Since $x$ lies to the left of $B$ and $f^t$ is originally non-decreasing, after updating values (before the permutation to $f^{t+1}$), the intermediate function $g^t$ remains non-decreasing to the left of $x$. Therefore, it is impossible to have two points $y_1 < y_2$ such that $\pi(y_1) > \pi(y_2)$ under our mapping rule. This proves that the set of upgrading points must be a single contiguous interval of the form $(x-\ell, x]$.

Next, consider the set of downgrading points. We first claim that it must start at $x+\delta$, since no point $y \le x+\delta$ can be a downgrading point. For any $y_2 \le x+\delta$, there is no point $y_1 < y_2$ that increases more than $y_2$ in this round. Formally, we have $g^t(y_1) - f^t(y_1) \le g^t(y_2) - f^t(y_2)$. Since $f^t$ is originally non-decreasing, and by our mapping rule, $y_2$ cannot be a downgrading point. 

For points $y > x+\delta$, $g^t$ remains non-decreasing in this range. Hence, it is also impossible to have two points $y_1 < y_2$ such that $\pi(y_1) > \pi(y_2)$ under the mapping rule. This shows that the set of downgrading points must be a single contiguous interval of the form $(x+\delta, x+\delta+\ell]$.
\end{proof}
Recall that we assume $x$ lies to the left of $B$. The value $F^{t+1}(x)$ is defined as the integral of $g^t$ over the set
\[
I' = (x-\ell, 1/2] \setminus (x+\delta, x+\delta+\ell].
\]
Thus, we have
\begin{align}
    F^{t+1}(x) 
    &= \int_{u \in I'} g^t(u) \, du 
    = \int_{u \in I'} f^t(u) \, du + \Gamma \nonumber \\
    &= F^t(x-\ell) - \int_{x+\delta}^{x+\delta+\ell} f^t(u) \, du + \Gamma, \label{eqn:f_t_1}
\end{align}
where $\Gamma$ captures the total discrepancy between $f^t$ and $g^t$ over $I'$, i.e.,
\[
\Gamma = \int_{u \in I'} g^t(u) \, du 
     - \int_{u \in I'} f^t(u) \, du.
\]
Since $f^t$ is non-decreasing, we have 
\[
\int_{x -\ell}^x f^t(u)\, du \;\le\; \int_{x+\delta}^{x+\delta+\ell} f^t(u)\, du.
\]
It follows that
\begin{align}
F^{t+1}(x)  
&=  F^t(x-\ell) - \int_{x+\delta}^{x+\delta+\ell} f^t(u) \, du + \Gamma \nonumber \\
&\le F^t(x) + \Gamma \nonumber  \\
&= \int_{x}^{x+\delta+\ell} f^t(u)\, du + \int_{x+\delta+\ell}^{1/2} f^t(u)\, du + \Gamma \nonumber \\
&\le \frac{\delta+\ell}{\ell} \cdot \int_{x+\delta}^{x+\delta+\ell} f^t(u)\, du + F^t(x+\ell+\delta) + \Gamma. \label{ineq:ub_of_F_t_1}
\end{align}

By multiplying \eqref{eqn:f_t_1} by $\delta+\ell$ and \eqref{ineq:ub_of_F_t_1} by $\ell$ and then summing, we obtain
\begin{align*}
(\delta + 2\ell)\, F^{t+1}(x) 
&\le (\delta + \ell)\, F^t(x-\ell) 
   + \ell\, F^t(x+\ell+\delta)
   + (\delta+2\ell)\, \Gamma.
\end{align*}
Applying the induction hypothesis on $F^t$, we have
\begin{align*}
(\delta + 2\ell)\, F^{t+1}(x) 
&\le (\delta + \ell)\!\left(\frac{k}{2} - 2k (x-\ell)^2\right) 
   + \ell\!\left(\frac{k}{2} - 2k (x+\ell+\delta)^2\right) 
   + (\delta+2\ell)\, \Gamma.
\end{align*}
A straightforward calculation then yields
\begin{align*}
(\delta+2\ell)\, F^{t+1}(x) 
&\le (\delta + 2\ell)\!\left(\frac{k}{2} - 2k x^2 - 2k \ell \cdot (\ell+\delta)\right)
   + (\delta+2\ell)\, \Gamma \\
&= (\delta + 2\ell)\!\left(\frac{k}{2} - 2k x^2 - 2k \ell \cdot (\ell+\delta) + \Gamma\right).
\end{align*}

It therefore suffices to show that $\Gamma \le 2k\ell\cdot(\ell+\delta)$. To this end, we analyze $\Gamma$ by examining how $I'$ intersects with $A$ and $B$. We establish the following relation:
\[
\Gamma 
= a \cdot \min(\ell + \delta, |A|) 
- b \cdot \max(|B| - \ell, 0) ~.
\]

The proof proceeds as follows. We first analyze the length of $A \cap I'$. Since the possible location of $A \cap I'$ lies within the interval $(x-\ell, x+\delta]$, and noting that $x+\delta$ is also the right endpoint of $A$, it follows directly that
$
|A \cap I'| = \min\{\ell + \delta, |A|\}.
$
The points in this interval contribute $a$ to $\Gamma$.
Next, we consider the length of $B \cap I'$, which lies within the interval $(x+\delta, x+\delta+\ell]$. Thus, its length is given by
$
|B \cap I'| = \min(|B|, \ell).
$
However, the points in $B \setminus I'$ effectively contribute $-b$ to $\Gamma$, yielding a total contribution of
$
- b \cdot \max(|B| - \ell, 0).
$

Finally, we discuss the four possible cases arising from the different outcomes of the $\min$ and $\max$ expressions.

\begin{itemize}
    \item $\ell + \delta \leq \abs{A}$ and $\abs{B} - \ell \geq 0$:
    In this case, we have
    \begin{align*}
        2k \ell \cdot (\ell+\delta)-\Gamma
        &= 2k\ell \cdot (\ell + \delta) - a\cdot (\ell + \delta) + b\cdot(\abs{B}-\ell)\\
        &= (2k\ell-a)\cdot (\ell + \delta)+b\cdot(\abs{B}-\ell)~. 
    \end{align*}
    We already have $2k \ell \cdot (\ell+\delta)-\Gamma \geq 0$, if
    $2k\ell-a\geq 0$. On the other hand, if $2k\ell-a<0$,  
    Define $r=\abs{A}-(\ell+\delta)\geq 0$,
    we have
    \begin{align*}
    2k\ell \cdot (\delta+\ell)-\Gamma
    &=  2k\ell \cdot(\abs{A}-r)-a\cdot (\ell+\delta)+b\cdot(\abs{B}-\ell)\\
    &=  2k\ell \cdot(\abs{A}-r)-a\cdot (\ell+\delta)+a\cdot \abs{A}-b\ell \tag{$a\cdot\abs{A}=b\cdot\abs{B}$}\\
    &\geq  2k\ell \cdot (\frac{b}{k(a+b)}-r)-a\cdot (\abs{A}-r)+a\cdot\abs{A}-b\ell\tag{$\abs{A}=\frac{b}{k(a+b)}\geq \frac{b}{2k}$}\\
    &\geq  2k\ell \cdot (\frac{b}{2k}-r)-a\cdot (\abs{A}-r)+a\cdot\abs{A}-b\ell\tag{$a+b \leq 2$}\\
    &= (a-2k\ell)\cdot r\geq 0~. \tag{$2k\ell-a<0$, $r\geq 0$}
\end{align*} 
This concludes the case.
    \item $\ell+\delta\leq \abs{A}$ and $\abs{B} - \ell < 0$: In this case, since $\ell>\abs{B}=\frac{a}{k(a+b)}\geq \frac{a}{2k}$ by $a + b \leq 2~$, 
    $$
        \Gamma = a\cdot (\ell+\delta) \leq 2k\ell \cdot (\ell+\delta)~.
    $$
    
    \item $\ell+\delta> \abs{A}$ and $\abs{B} - \ell \geq 0$: In this case, since $\ell+\delta>\abs{A}=\frac{b}{k(a+b)}\geq\frac{b}{2k}$ by $a+b \leq 2~$, 
    $$
         \Gamma = a\abs{A}-b\abs{B}+b\ell = b\ell \leq 2k\ell\cdot(\ell+\delta)~.
    $$
    \item $\ell + \delta > |A|$ and $|B| - \ell < 0$: 
    We show that this case is impossible. Under these conditions, we can find an upgrading point $x_1 \in (x-\ell,x]$ such that $x_1 \notin A$, and a downgrading point $x_2 \in (x+\delta,x+\delta+\ell]$ such that $x_2 \notin B$. Therefore, we have $g^t(x_1) \le g^t(x_2)$ and $x_2 > x_1$. By the mapping rule, it must hold that $\pi(x_1) \le \pi(x_2)$, contradicting the requirement that $x_1$ is mapped to a value strictly larger than $x_2$. Hence, this case cannot occur.

\end{itemize}
Finally, it remains to show that $F^{t+1}(x) \le \frac{k}{2} - 2k x^2$ holds for any point $x$ to the right of $A$. We demonstrate that this case can be reduced to the one we have already completed.

Construct a function $\tilde{f}^t(y) = -f^t(-y)$, which is still non-decreasing over $I$. Define its integral by
$
\tilde{F}^t(y) = \int_{y}^{\frac12} \tilde{f}^t(u) \, du.
$
Then we have
\begin{align*}
\tilde{F}^t(y) 
&= \int_{y}^{\frac12} \tilde{f}^t(u) \, du 
= \int_{y}^{\frac12} (-f^t(-u)) \, du 
= -\int_{-\frac12}^{-y} f^t(v) \, dv \\
&= F^t(-y) - F^t\!\!\left(-\frac12\right)
= F^t(-y) ~, 
\tag{by \Cref{inv:stacking_f}}
\end{align*}
where the last equality follows from our integral normalization.

Correspondingly, we apply the operation $(a', b', A', B')$ on $g^t(\cdot)$ by setting $A' = -B$, $B' = -A$, $a' = b$, and $b' = a$, which preserves the property that
$
\tilde{F}^{t+1}(y) = F^{t+1}(-y).
$ 
Since $x \ge y$ for any $y \in A$, we have $-x \le -y$ for all $-y \in B'$. Thus, $-x$ is to the left of $B'$. Applying the initial analysis of our proof, we obtain
\[
\tilde{F}^{t+1}(-x) 
\le \frac{k}{2} - 2k(-x)^2 
= \frac{k}{2} - 2k x^2~.
\]
Since $\tilde{F}^{t+1}(-x) = F^{t+1}(x)$, it follows that
$
F^{t+1}(x) \le \frac{k}{2} - 2k x^2,
$
which completes the proof.
Finally, we note that under this symmetric construction, every left-open right-closed subinterval may become right-open left-closed. However, this does not affect any step of the argument.
\end{proof}

%% file: figure/stacking-example1.tex
\begin{tikzpicture}[scale=0.8]
\draw[->] (-2.3,0) -- (2.3,0) node[right] {\scriptsize $x$};
\draw[->] (0,-1.5) -- (0,1.5) node[right] {\scriptsize $f^0$};

\draw (-2, 0) node[below] {\scriptsize $-\frac12$};
\draw (2, 0) node[below] {\scriptsize $\frac12$};

\draw[-, semithick, blue] (-2,0) -- (2,0);

\filldraw[black] (-2, 0) circle (2pt);
\filldraw[black] (0, 0) circle (2pt);
\filldraw[black] (2, 0) circle (2pt);

\draw[->, ultra thick,  draw=green!60!black] (-1.5,0.1) -- (-1.5,0.5);
\draw[->, ultra thick,  draw=green!60!black] (-0.5,0.1) -- (-0.5,0.5);
\draw[->, ultra thick,  draw=red!60] (0.5,-0.1) -- (0.5,-0.5);
\draw[->, ultra thick,  draw=red!60] (1.5,-0.1) -- (1.5,-0.5);

\pgfmathsetmacro{\shiftx}{8};
\draw[->] (3.5,0.5) -- (4.5,0.5) 
    node[midway, above] {\scriptsize $\substack{A = (-1/2, 0], \\ B = (-0,1/2]}$};

\draw[->] (-2.3+\shiftx,0) -- (2.3+\shiftx,0) node[right] {\scriptsize $x$};
\draw[->] (0+\shiftx,-1.5) -- (0+\shiftx,1.5) node[right] {\scriptsize $g^0$};

\draw[-, semithick, blue] (\shiftx, -0.8) -- (2+\shiftx, -0.8);
\draw[-, semithick, blue] (\shiftx, -0.8) -- (0+\shiftx, 0.8);
\draw[-, semithick, blue] (\shiftx, 0.8) -- (-2+\shiftx, 0.8);

\draw (-2+\shiftx, 0) node[below] {\scriptsize $-\frac12$};
\draw (2+\shiftx, 0) node[below] {\scriptsize $\frac12$};

\filldraw[black] (-2+\shiftx, 0) circle (2pt);
\filldraw[black] (0+\shiftx, 0) circle (2pt);
\filldraw[black] (2+\shiftx, 0) circle (2pt);

\draw (0.1+\shiftx, -0.8) node[left] {\scriptsize $-1$};
\draw (0+\shiftx, 0.8) node[right] {\scriptsize $1$};
\draw[dotted, -, semithick] (0+\shiftx,1) -- (\shiftx,1); 

\draw[->] (11,0.5) -- (12,0.5) 
    node[midway, above] {\scriptsize \shortstack{permuting}};

\pgfmathsetmacro{\shiftx}{15};

\draw[->] (-2.3+\shiftx,0) -- (2.3+\shiftx,0) node[right] {\scriptsize $x$};
\draw[->] (0+\shiftx,-1.5) -- (0+\shiftx,1.5) node[right] {\scriptsize $f^1$};

\draw[-, semithick, blue] (\shiftx, 0.8) -- (2+\shiftx, 0.8);
\draw[-, semithick, blue] (\shiftx, 0.8) -- (0+\shiftx, -0.8);
\draw[-, semithick, blue] (\shiftx, -0.8) -- (-2+\shiftx, -0.8);

\draw (-2+\shiftx, 0) node[below] {\scriptsize $-\frac12$};
\draw (2+\shiftx, 0) node[below] {\scriptsize $\frac12$};

\filldraw[black] (-2+\shiftx, 0) circle (2pt);
\filldraw[black] (0+\shiftx, 0) circle (2pt);
\filldraw[black] (2+\shiftx, 0) circle (2pt);

\draw (0.1+\shiftx, -0.8) node[right] {\scriptsize $-1$};
\draw (0+\shiftx, 0.8) node[left] {\scriptsize $1$};
\draw[dotted, -, semithick] (0+\shiftx,1) -- (\shiftx,1); 

\end{tikzpicture}

%% file: figure/algo-stacking.tex
\begin{tikzpicture}
\draw[->] (-3.3, 0) -- (3.3,0) node[right] {$x$};
\draw[->] (0,-1.2) -- (0,1.2) node[above] {$f^0$};

\foreach \x in {-1.5, -0.5, 0, 0.5, 1.5}
\draw (\x,0) node[below] {};
\draw (-2.5, 0) node[below] {\scriptsize $H_1^1$};
\draw (-1.5, 0) node[below] {\scriptsize $H_1^2$};
\draw (-0.5, 0) node[below] {\scriptsize $H_2^1$};
\draw (0.5, 0) node[below] {\scriptsize $H_2^2$};
\draw (1.5, 0) node[below] {\scriptsize $H_3^1$};
\draw (2.5, 0) node[below] {\scriptsize $H_3^2$};

\draw[-, semithick, blue] (-3, 0) -- 
(3, 0);

\filldraw[black] (-3, 0) circle (2pt);
\filldraw[black] (-2, 0) circle (2pt);
\filldraw[black] (-1, 0) circle (2pt);
\filldraw[black] (0, 0) circle (2pt);
\filldraw[black] (1, 0) circle (2pt);
\filldraw[black] (2, 0) circle (2pt);
\filldraw[black] (3, 0) circle (2pt);

\draw[->, ultra thick,  draw=green!60!black] (-2.5,0.1) -- (-2.5,0.5);
\draw[->, ultra thick,  draw=red!60] (-0.5,-0.4) -- (-0.5,-0.8);
\draw[->, ultra thick,  draw=red!60] (1.5,-0.4) -- (1.5,-0.8);

\pgfmathsetmacro{\shiftx}{8};
\draw[->] (3,1) -- (5.5,1) 
    node[midway, above] {\scriptsize allocate item $1$ to agent 1};

\draw[->] (-3.3+\shiftx, 0) -- (3.3+\shiftx,0) node[right] {$x$};
\draw[->] (0+\shiftx,-1.2) -- (0+\shiftx,1.2) node[above] {$f^1$};

\draw (-2.5+\shiftx, 0) node[below] {\scriptsize $H_2^1$};
\draw (-1.5+\shiftx, 0) node[below] {\scriptsize $H_3^1$};
\draw (-0.5+\shiftx, 0) node[below] {\scriptsize $H_1^2$};
\draw (0.5+\shiftx, 0) node[below] {\scriptsize $H_2^2$};
\draw (1.5+\shiftx, 0) node[below] {\scriptsize $H_3^2$};
\draw (2.5+\shiftx, 0) node[below] {\scriptsize $H_1^1$};

\draw[-, dotted, thick] (0+\shiftx, 1) -- (3+\shiftx, 1);
\draw[-, dotted, thick] (0+\shiftx, -0.5) -- (-2+\shiftx, -0.5);

\draw[-, semithick, blue] (-3+\shiftx, -0.5) -- 
(-1+\shiftx, -0.5);
\draw[-, semithick, blue] (-1+\shiftx, 0) -- 
(-1+\shiftx, -0.5);
\draw[-, semithick, blue] (-1+\shiftx, 0) -- 
(2+\shiftx, 0);
\draw[-, semithick, blue] (2+\shiftx, 0) -- 
(2+\shiftx, 1);
\draw[-, semithick, blue] (2+\shiftx, 1) -- 
(3+\shiftx, 1);

\filldraw[black] (-3+\shiftx, 0) circle (2pt);
\filldraw[black] (-2+\shiftx, 0) circle (2pt);
\filldraw[black] (-1+\shiftx, 0) circle (2pt);
\filldraw[black] (0+\shiftx, 0) circle (2pt);
\filldraw[black] (1+\shiftx, 0) circle (2pt);
\filldraw[black] (2+\shiftx, 0) circle (2pt);
\filldraw[black] (3+\shiftx, 0) circle (2pt);

\draw (0+\shiftx, 1) node[left] {\scriptsize $1$};
\draw (0+\shiftx, -0.5) node[right] {\scriptsize $-\frac12$};

\end{tikzpicture}

%% file: figure/idea_of_contiguous.tex
\begin{tikzpicture}
    \centering
    \draw[->, semithick] (-3.2,0) -- (3.2,0) node[right, anchor=north] {\scriptsize $x$};
    \draw[->, semithick] (0,-1) -- (0,1.5) node[above] {\scriptsize $f^t$};
    
    \draw (-3, 0) node[below] {\scriptsize $-\frac12$};
    \draw (3, 0) node[below] {\scriptsize $\frac12$};

    \draw[-, black] (-3, -0.6) -- (0, -0.6);
    \draw[-, black] (0, 0.6) -- (3, 0.6);
    
    \draw[-, very thick, green!60!black] (1.5, 0.6) -- (3, 0.6);
    \draw[-, very thick, red] (0.75, 0.6) -- (1.5, 0.6);
    \draw[-, very thick, red] (-0.75, -0.6) -- (0, -0.6);

    \draw[->, ultra thick,  draw=green!60!black] (2,0.5) -- (2,0.1);
    \draw[->, ultra thick,  draw=green!60!black] (2.5,0.5) -- (2.5,0.1);
    \draw[->, ultra thick,  draw=red] (1.125,0.7) -- (1.125,1.1);
    \draw[->, ultra thick,  draw=red] (-0.375,-0.5) -- (-0.375,-0.1);

\pgfmathsetmacro{\shift}{8};
\centering
\draw[->, semithick] (-3.2+\shift,0) -- (3.2+\shift,0) node[right, anchor=north] {\scriptsize $x$};
\draw[->, semithick] (0+\shift,-1) -- (0+\shift,1.5) node[above] {\scriptsize $f^t$};

\draw (-3+\shift, 0) node[below] {\scriptsize $-\frac12$};
\draw (3+\shift, 0) node[below] {\scriptsize $\frac12$};

\draw[-, black] (-3+\shift, -0.6) -- (0+\shift, -0.6);
\draw[-, black] (0+\shift, 0.6) -- (3+\shift, 0.6);

\draw[-, very thick, green!60!black] (1.5+\shift, 0.6) -- (3+\shift, 0.6);
\draw[-, very thick, red] (0+\shift, 0.6) -- (1.5+\shift, 0.6);

\draw[->, ultra thick,  draw=green!60!black] (2+\shift,0.5) -- (2+\shift,0.1);
\draw[->, ultra thick,  draw=green!60!black] (2.5+\shift,0.5) -- (2.5+\shift,0.1);
\draw[->, ultra thick,  draw=red] (1+\shift,0.7) -- (1+\shift,1.1);
\draw[->, ultra thick,  draw=red] (0.5+\shift,0.7) -- (0.5+\shift,1.1);
\end{tikzpicture}

\begin{tikzpicture}
    \centering
    \draw[->, semithick] (-3.2,0) -- (3.2,0) node[right, anchor=north] {\scriptsize $x$};
    \draw[->, semithick] (0,-1) -- (0,1.5) node[above] {\scriptsize $f^{t+1}$};
    
    \draw (-3, 0) node[below] {\scriptsize $-\frac12$};
    \draw (3, 0) node[below] {\scriptsize $\frac12$};

    \draw[-, black] (-3, -0.6) -- (-0.75, -0.6);
    \draw[-, black] (-0.75, -0.6) -- (-0.75, 0);
    \draw[-, black] (-0.75, 0) -- (1.5, 0);
    \draw[-, black] (1.5, 0) -- (1.5, 0.6);
    \draw[-, black] (1.5, 0.6) -- (2.25, 0.6);
    \draw[-, black] (2.25, 0.6) -- (2.25, 1.2);
    \draw[-, black] (2.25, 1.2) -- (3, 1.2);

    \fill[blue!20] (2, 0) rectangle (2.25, 0.6);
    \fill[blue!20] (2.25, 0) rectangle (3, 1.2);
    \draw[-, dotted, thick, blue] (2, 0.6) -- (2, 0);
    \draw (2.62, 0.6) node {\scalebox{0.6}{$F^{t+1}(y)$}};
    \draw (2, 0) node[below] {\scriptsize $y$};

\pgfmathsetmacro{\shift}{8};
\centering
\draw[->, semithick] (-3.2+\shift,0) -- (3.2+\shift,0) node[right, anchor=north] {\scriptsize $x$};
\draw[->, semithick] (0+\shift,-1) -- (0+\shift,1.5) node[above] {\scriptsize $\tilde{f}^{t+1}$};

\draw (-3+\shift, 0) node[below] {\scriptsize $-\frac12$};
\draw (3+\shift, 0) node[below] {\scriptsize $\frac12$};

\draw[-, black] (-3+\shift, -0.6) -- (0+\shift, -0.6);
\draw[-, black] (0+\shift, 0) -- (1.5+\shift, 0);
\draw[-, black] (1.5+\shift, 0) -- (1.5+\shift, 1.2);
\draw[-, black] (1.5+\shift, 1.2) -- (3+\shift, 1.2);

\fill[blue!20] (2+\shift, 0) rectangle (3+\shift, 1.2);
\draw[-, dotted, thick, blue] (2+\shift, 1.2) -- (2+\shift, 0);
\draw (2.5+\shift, 0.6) node {\scalebox{0.6}{$\tilde{F}^{t+1}(y)$}};
\draw (2+\shift, 0) node[below] {\scriptsize $y$};
\end{tikzpicture}

%% file: figure/induction-pic1.tex
\begin{tikzpicture}
\centering
\draw[->, semithick] (-5,0) -- (5,0) node[right, anchor=north] {\scriptsize $x$};
\draw[->, semithick] (0,-1) -- (0,1.5) node[above] {\scriptsize $f^t$};

\draw (-4.5, 0) node[below] {\scriptsize $-\frac12$};
\draw (4.5, 0) node[below] {\scriptsize $\frac12$};
\draw (0.5, 0) node[below] {\small $x$};

\draw[-, very thick, green!60!black] (-1, 0) -- (2, 0);
\draw[-, very thick, orange] (3.5, 0) -- (4.5, 0);
\draw[-, very thick, red] (2, 0) -- (3.5, 0);
\draw[-, very thick, blue] (-4.5, 0) -- (-1, 0);

\foreach \x in {0.5, -1, -4.5, 2, 3.5, 4.5} {
    \draw[black, line width=0.8pt] (\x, -0.1) -- (\x, 0.1);
  }

\draw[dotted, black] (-4.5, 0) --(-4.5, -1);
\draw[-, semithick, black] (-4.5, -1) --(-2.5, -1);
\draw[-, semithick, black] (-2.5, -0.6) --(-2.5, -1);
\draw[-, semithick, black] (-2.5, -0.6) --(-1, -0.6);
\draw[-, semithick, black] (-1, -0.6) --(-1, 0.3);
\draw[-, semithick, black] (-1, 0.3) --(2, 0.3);
\draw[-, semithick, black] (2, 0.3) --(2, 0.5);
\draw[-, semithick, black] (2, 0.5) --(3.5, 0.5);
\draw[-, semithick, black] (3.5, 0.5) --(3.5, 1);
\draw[-, semithick, black] (3.5, 1) --(4.5, 1);
\draw[dotted, black] (4.5, 0) --(4.5, 1);

\draw[->, ultra thick,  draw=green!60!black] (-0.25,0.4) -- (-0.25,0.7);
\draw[->, ultra thick,  draw=green!60!black] (1.25,0.4) -- (1.25,0.7);

\draw[->, ultra thick,  draw=red] (2.5,0.4) -- (2.5,0.1);
\draw[->, ultra thick,  draw=red] (3,0.4) -- (3,0.1);
\draw[dotted, -, thick] (0.5, -1) -- (0.5,1.5);


\end{tikzpicture}

%% file: figure/induction-pic2.tex
\begin{tikzpicture}
\centering
\draw[->, semithick] (-5,0) -- (5,0) node[right, anchor=north] {\scriptsize $x$};
\draw[->, semithick] (0,-1) -- (0,1.5) node[above] {\scriptsize $f^t$};

\draw (-4.5, 0) node[below] {\scriptsize $-\frac12$};
\draw (4.5, 0) node[below] {\scriptsize $\frac12$};
\draw (0.5, 0) node[below] {\small $x$};

\draw[-, very thick, green!60!black] (-0.5, 0) -- (0, 0);
\draw[-, very thick, green!60!black, decorate, decoration={zigzag}] (-1, 0) -- (0.5, 0);
\draw[-, very thick, green!60!black] (0.5, 0) -- (2, 0);
\draw[-, very thick, orange, decorate, decoration={zigzag}] (3.5, 0) -- (4, 0);
\draw[-, very thick, orange] (4, 0) -- (4.5, 0);
\draw[-, very thick, red, decorate, decoration={zigzag}] (2, 0) -- (3.5, 0);
\draw[-, very thick, blue] (-4.5, 0) -- (-1.5, 0);
\draw[-, very thick, blue, decorate, decoration={zigzag}] (-1.5, 0) -- (-1, 0);

\foreach \x in {0.5, -1, -4.5, 2, 3.5, 4.5} {
    \draw[black, line width=0.8pt] (\x, -0.1) -- (\x, 0.1);
  }

\draw[dotted, black] (-4.5, 0) --(-4.5, -1);
\draw[-, semithick, black] (-4.5, -1) --(-2.5, -1);
\draw[-, semithick, black] (-2.5, -0.6) --(-2.5, -1);
\draw[-, semithick, black] (-2.5, -0.6) --(-1, -0.6);
\draw[-, semithick, black] (-1, -0.6) --(-1, 0.3);
\draw[-, semithick, black] (-1, 0.3) --(2, 0.3);
\draw[-, semithick, black] (2, 0.3) --(2, 0.5);
\draw[-, semithick, black] (2, 0.5) --(3.5, 0.5);
\draw[-, semithick, black] (3.5, 0.5) --(3.5, 1);
\draw[-, semithick, black] (3.5, 1) --(4.5, 1);
\draw[dotted, black] (4.5, 0) --(4.5, 1);

\draw[->, ultra thick,  draw=green!60!black] (-0.25,0.4) -- (-0.25,0.7);
\draw[->, ultra thick,  draw=green!60!black] (1.25,0.4) -- (1.25,0.7);

\draw[->, ultra thick,  draw=red] (2.5,0.4) -- (2.5,0.1);
\draw[->, ultra thick,  draw=red] (3,0.4) -- (3,0.1);
\draw[dotted, -, thick] (0.5, -1) -- (0.5,1.5);
\draw[dotted, -, thick] (2, -1) -- (2,0);
\draw[dotted, -, thick] (-1.5, -1) -- (-1.5,0);
\draw[dotted, -, thick] (4, -1) -- (4,0);


\draw[decorate, gray, thick, decoration={brace, amplitude=5pt}, yshift=-10pt]
    (0.5,0) -- (-1.5,0) node[below, midway, yshift=-8pt] {\small $\ell$};
\draw[decorate, gray, thick, decoration={brace, amplitude=5pt}, yshift=-10pt]
    (2,0) -- (0.5,0) node[below, midway, yshift=-8pt] {\small $\delta$};
   \draw[decorate, gray, thick, decoration={brace, amplitude=5pt}, yshift=-10pt]
    (4,0) -- (2,0) node[below, midway, yshift=-8pt] {\small $\ell$}; 
\end{tikzpicture}

%% file: 50-negative-warm.tex
\section{Negative Results}
\label{sec:negative}
In this section, we study general additive disutility functions, with unbounded $k$ and $D$. 
As we have mentioned in the introduction, we can easily achieve a competitive ratio of $n$ of MMS by any allocation (notice that $\frac1nd_i(M)$ is a lower bound to $\MMS_i$ and the disutility for receiving all items is $d_i(M)$).
The following theorem shows that this is the best we can do.

\begin{theorem}\label{thm:lowerbound}
    For any $n$ and $\varepsilon>0$, no online algorithm can guarantee outputting $(n-\varepsilon)$-MMS allocations in the $n$-agent online MMS allocation problem for indivisible chores. 
\end{theorem}

We will prove Theorem~\ref{thm:lowerbound} in the remaining part of this section.
We will first prove the theorem for $n=2$ in \Cref{sect:negative-2} for a warm-up.
The proof contains many of the key intuitions which is used for proving Theorem~\ref{thm:lowerbound}.
The full proof of Theorem~\ref{thm:lowerbound} is then presented in \Cref{sect:negative-general}.

\subsection{Warm-Up: the Two-Agent Special Case}
\label{sect:negative-2}
In this section, we prove Theorem~\ref{thm:lowerbound} for $n=2$ by providing a tight example of $(2-\varepsilon)$.
After this, we will recap the key insights behind the proof, which will be useful for proving the general version of Theorem~\ref{thm:lowerbound}.

It is natural to think about the problem where there is an ``adversary'' who, upon seeing the algorithm's output allocation for the first $j$ items, adaptively gives the disutilities of the next item for both agents.
Consider an arbitrary fixed online algorithm.
We will describe the tight example by defining $d_1(j+1)$ and $d_2(j+1)$ given the allocation of $[j]$, and we will show that we will eventually have $d_i(A_i)>(2-\varepsilon)\MMS_i$ for some $i\in\{1,2\}$ at some iteration $j^\ast$.\footnote{Whenever we write agent $i$'s allocated bundle $A_i$ and agent $i$'s MMS threshold $\MMS_i$, it is with respect to the set of items that have been allocated, i.e., it is with respect to the item set $[j]$ if we are analyzing the first $j$ iterations of the algorithm. Thus, it is more formal to define notations such as $A_i[j]$ and $\MMS_i[j]$ to be explicit about the item set $[j]$. However, we choose not to do so for notational simplicity, and we find that the value $j$ we are concerned with is always clear from the context in our discussions in this section.}

Before we present the formal proof, we first describe the high-level ideas behind the construction of $d_1$ and $d_2$.
We will define $d_1$ such that agent $1$ is forbidden to take items in two consecutive iterations.
Specifically, we mean that $d_1(A_1)$ will be close to $2\cdot\MMS_1$ if agent $1$ takes items in two consecutive iterations.
Such a $d_i$ can be constructed as follows: if item $j$ is not taken by agent $1$, then the disutility of the next item $j+1$ for agent $1$ is set to a very high value such that $d_1([j])$ is negligible compared with $d_1(j+1)$; otherwise, the next item $j+1$ has the same disutility as that of item $j$.
Therefore, at this stage, we know that agent $2$ has to take at least one item in every two consecutive iterations, regardless of $d_2$.
We next construct $d_2$ such that agent $2$ will eventually reach a stage where $d_2(A_2)\approx2\cdot\MMS_2$.
This can be done as follows.
By keeping letting the next item $j+1$ have a disutility significantly larger than the disutility of the item set $[j]$, we can construct $d_2$ such that the first item agent $2$ takes has a very large disutility, such that the disutility for all previous items is negligible.
After agent $2$ takes her first item, we will construct $d_2$ such that 1) $d_2(j+1)$ is the same as the disutility of the first item taken by agent $2$ if item $j$ was taken by agent $1$, and 2) $d_2(j+1)$ is small if item $j$ was not taken by agent $1$. 
Given that agent $1$ cannot take items in two consecutive iterations, item $(j+1)$ must be taken by agent $2$ if agent $1$ takes $j$.
In this case, agent $2$ has taken two ``large items'': the first item she receives and the item $(j+1)$, and all the other items have small disutilities.
This implies $d_2(A_2)\approx2\cdot\MMS_2$.
On the other hand, if agent $1$ never takes any more items after agent $2$ takes her first item, agent $2$ will keep receiving items with small disutilities, and eventually agent $2$'s disutility will reach approximately $2\cdot\MMS_2$.

Now we formally prove Theorem~\ref{thm:lowerbound} for $n=2$.
In the following, we let $\epsilon_1=\frac\varepsilon2$ and $\epsilon_2$ be a small rational number such that $\epsilon_2\leq \frac\varepsilon3$ and $1/\epsilon_2$ is an integer.
For a positive integer $w$, let $j_i^{(w)}$ be the iteration where it is the $w$-th time that agent $i$ has taken an item.
By our definition of item index, $j_i^{(w)}$ is also the index of this item.
For example, $j_2^{(1)}$ is the first iteration where agent $2$ takes an item.
Notice that, at this stage, it is possible that $j_i^{(w)}$ may not exist, as the algorithm may decide to let the other agent take items forever in the future after agent $i$ has taken her $(w-1)$-th item.

For the first item $1$, we define $d_1(1)=d_2(1)=1$.
Agent $1$'s disutility $d_1(j+1)$ for $j=1,2,\ldots$ is defined as
$$d_1(j+1)=\left\{\begin{array}{ll}
    d_1(j) & \mbox{if agent 1 has taken item }j \\
    d_1(j)/\epsilon_1 & \mbox{otherwise}
\end{array}\right..$$
Agent $2$'s disutility is defined as
$$d_2(j+1)=\left\{\begin{array}{ll}
    d_2([j])/\epsilon_2 & \mbox{if agent 2 has not taken any item yet} \\
    \epsilon_2\cdot d_2(j_2^{(1)}) & \mbox{if agent 2 has taken item }j\\
    d_2(j_2^{(1)}) & \mbox{if agent 2 has not taken item }j\mbox{, and }j\geq j_2^{(1)}
\end{array}\right..$$

Next, we will show that we will have either $d_1(A_1)> (2-\varepsilon)\cdot\MMS_1$ or $d_2(A_2)>(2-\varepsilon)\cdot\MMS_2$ after finitely many iterations.
We first observe that we will have $d_1(A_1)>(2-\varepsilon)\cdot\MMS_1$ if agent $1$ has taken the items in the two consecutive iterations $j-1$ and $j$.
If this is the case, by our definition, we have $d_1(j)=d_1(j-1)$ and $d_1([j-2])=\epsilon_1\cdot d_1(j)$.
At the end of iteration $j$, we have $\MMS_1\leq (1+\epsilon_1)\cdot d_1(j)$ (by noticing that $([j-1],\{j\})$ is a bi-partition where the bundle $[j-1]$ has the maximum disutility $(1+\epsilon_1)\cdot d_1(j)$).
On the other hand, we have
$$d_1(A_1)\geq d_1(\{j-1,j\})=2d_1(j)>(2-\varepsilon)\cdot (1+\epsilon_1)d_1(j)\geq (2-\varepsilon)\cdot\MMS_1,$$
which proves our claim.
We will assume agent $1$ never takes items in two consecutive iterations from now on.
Notice that this would imply the existence of $j_2^{(w)}$ for every $w$, and in fact we must have $j_2^{(w)}\leq 2w$.

Lastly, we will conclude the proof of Theorem~\ref{thm:lowerbound} for $n=2$ by showing that $d_2(A_2)>(2-\varepsilon)\cdot\MMS_2$ after finitely many iterations.
We discuss two cases regarding whether agent $1$ has taken an item after iteration $j_2^{(1)}$.

If agent $1$ has not taken an item after iteration $j_2^{(1)}$, we prove that $d_2(A_2)>(2-\varepsilon)\cdot\MMS_2$ at iteration $j_2^{(1)}+1/\epsilon_2$.
The discussions below in this paragraph are with respect to the item set $[j_2^{(1)}+1/\epsilon_2]$ corresponding to the first $j_2^{(1)}+1/\epsilon_2$ iterations.
On the one hand, we must have $\MMS_2\leq (1+\epsilon_2)\cdot d_2(j_2^{(1)})$.
To see this, consider the bi-partition $([j_2^{(1)}],[j_2^{(1)}+1/\epsilon_2]\setminus [j_2^{(1)}])$.
By our definition of $d_2(\cdot)$, we have $d_2([j_2^{(1)}])=(1+\epsilon_2)\cdot d_2(j_2^{(1)})$ and $d_2([j_2^{(1)}+1/\epsilon_2]\setminus [j_2^{(1)}])=(1/\epsilon_2)\cdot \epsilon_2d_2(j_2^{(1)})=d_2(j_2^{(1)})$.
Thus, $\MMS_2\leq (1+\epsilon_2)d_2(j_2^{(1)})$.
On the other hand, we have
$$d_2(A_2)\geq d_2\left(\{j_2^{(1)}\}\cup \left([j_2^{(1)}+1/\epsilon_2]\setminus [j_2^{(1)}]\right) \right)=2d_2(j_2^{(1)})$$
$$\qquad>(2-\varepsilon)(1+\epsilon_2)d_2(j_2^{(1)})\geq(2-\varepsilon)\cdot\MMS_2, $$
concluding the proof for this case.

If agent $1$ has taken an item after iteration $j_2^{(1)}$, let $j^\ast$ be this iteration (notice that $j^\ast$ can only be either $j_1^{(1)}$ or $j_1^{(2)}$ given that agent $1$ cannot take items in two consecutive iterations, and the value of $j^\ast$ depends on if agent $1$ has taken an item before $j_2^{(1)}$).
We prove that $d_2(A_2)>(2-\varepsilon)\cdot\MMS_2$ by the end of iteration $j^\ast+1$.
The following discussions are with respect to the item set $[j^\ast+1]$ corresponding to the first $j^\ast+1$ iterations.
First, we note that, among the items allocated between iteration $j_2^{(1)}$ and iteration $j^\ast$, agent $2$ has taken $j_2^{(1)},j_2^{(1)}+1,\ldots,j^\ast-1$, and agent $1$ has taken $j^\ast$.
In addition, agent $2$ must take $j^\ast+1$ given that agent $1$ cannot take items in two consecutive iterations.
Thus, we have
\begin{align*}
    d_2(A_2)&\geq d_2\left(\{j_2^{(1)}\}\cup\left([j^\ast-1]\setminus [j_2^{(1)}]\right)\cup\{j^\ast+1\}\right)\\
    &=d_2(j_2^{(1)})+(j^\ast-1-j_2^{(1)})\cdot \epsilon_2d_2(j_2^{(1)})+d_2(j_2^{(1)}))\\
    &=\left(2+(j^\ast-1-j_2^{(1)})\epsilon_2\right)d_2(j_2^{(1)}).
\end{align*}
Next, we find an upper bound to $\MMS_2$ by considering the bi-partition $(B_1,B_2)$ of $[j^\ast+1]$ where $B_1=[j_2^{(1)}+\lfloor\frac{j^\ast-j_2^{(1)}}2\rfloor]$ and $B_2=[j^\ast+1]\setminus B_1$.
By our definition of $d_2$, we have
$d_2([j_2^{(1)}-1])\leq \epsilon_2\cdot d_2(j_2^{(1)})$ with equality holds if and only if $j_2^{(1)}>1$.
Thus, $d_2(B_1)\leq \epsilon_2d_2(j_2^{(1)})+d_2(j_2^{(1)})+\lfloor\frac{j^\ast-j_2^{(1)}}2\rfloor\cdot\epsilon_2 d_2(j_2^{(1)})=(1+(\lfloor\frac{j^\ast-j_2^{(1)}}2\rfloor+1)\epsilon_2)d_2(j_2^{(1)})$, and $d_2(B_2)=d_2(j^\ast+1)+d_2([j^\ast]\setminus [j_2^{(1)}+\lfloor\frac{j^\ast-j_2^{(1)}}2\rfloor])=d_2(j_2^{(1)})+\lceil\frac{j^\ast-j_2^{(1)}}2\rceil\cdot\epsilon_2d_2(j_2^{(1)})\leq (1+(\lfloor\frac{j^\ast-j_2^{(1)}}2\rfloor+1)\epsilon_2)d_2(j_2^{(1)})$.
Therefore,
$$\MMS_2\leq \left(1+\left(\left\lfloor\frac{j^\ast-j_2^{(1)}}2\right\rfloor+1\right)\epsilon_2\right)d_2\left(j_2^{(1)}\right)\leq \left(1+\left(\frac{j^\ast-j_2^{(1)}}2+1\right)\epsilon_2\right)d_2\left(j_2^{(1)}\right).$$
Putting together,
\begin{align*}
    d_2(A_2)-(2-\varepsilon)\MMS_2&\geq\left(\left(2+(j^\ast-1-j_2^{(1)})\epsilon_2\right)-(2-\varepsilon)\left(1+\left(\frac{j^\ast-j_2^{(1)}}2+1\right)\epsilon_2\right)\right)\cdot d_2\left(j_2^{(1)}\right)\\
    &=\left(\varepsilon-3\varepsilon_2+\varepsilon\varepsilon_2\left(\frac{j^\ast-j_2^{(1)}}2+1\right)\right)\cdot d_2\left(j_2^{(1)}\right)\\
    &>(\varepsilon-3\epsilon_2)\cdot d_2\left(j_2^{(1)}\right)\\
    &\geq0,
\end{align*}
which concludes the proof.

\paragraph{Key insights in the proof, and on generalization to more agents.}
In the proof above, we have constructed agent $1$'s disutility function such that agent $1$ can never take items in two consecutive iterations.
This ensures that agent $2$ must take items ``frequently enough''.
Specifically, agent $2$ must take at least one item in every two consecutive iterations.
Agent $2$'s disutility function is then defined by exploiting this.

Now, suppose there are $n$ agents and the algorithm only allocates items among agents $1$ and $2$.
Let $\MMS_1$ and $\MMS_2$ be defined with $n$-partitions (instead of bi-partitions in the proof above).
By the same construction of $d_1(\cdot)$, we can see that agent $1$ cannot take items in $n$ consecutive iterations to avoid a disutility of $(n-\varepsilon)\cdot\MMS_1$.
Therefore, agent $2$ must take at least one item in every $n$ consecutive iterations.
We exploit this fact and define $d_2$ such that agent $2$'s disutility will eventually reach approximately $n\cdot\MMS_2$.
By using a similar construction in the proof above, we can make sure agent $2$'s disutility on the first item she receives is approximately $\MMS_2$ by setting the disutilities of all previous items to be small before agent $2$ takes her first item.
To exploiting that agent $2$ must take an item after at most $n-1$ iterations where agent $1$ has been taking items, we construct an increasing sequence of numbers $\{a_1,a_2,a_3,\ldots\}$ such that $a_{n+1}$ is agent $2$'s disutility for the first item she receives, i.e., $a_{n+1}=d_2(j_2^{(1)})$, and, for each $t$, $a_{t}$ is significantly larger than $\sum_{t'=1}^{t-1} a_{t'}$.
After agent $2$ takes her first item, we set the disutility of the next item for agent $2$ to $a_1$.
In each future iteration $j$, the disutility $d_2(j)$ depends on both $d_2(j-1)$ and how item $j-1$ was allocated: if agent $2$ has taken $j-1$, the disutility of $j$ is set to $a_1$; otherwise, set $d_2(j)=a_{t+1}$ where $a_t=d_2(j-1)$.
By defining $d_2(\cdot)$ in this way, we can ensure that, between every two iterations $j_2^{(w)}$ and $j_2^{(w+1)}$ where agent $2$ takes items, the items allocated to agent $1$ from iteration $j_2^{(w)}+1$ to iteration $j_2^{(w+1)}-1$ have a negligible overall disutility to agent $2$ compared with $j_2^{(w+1)}$.
Finally, by the time agent $2$'s disutility $d_2(A_2)$ becomes $n\cdot d_2(j_2^{(1)})$, the items taken by agent $1$ will still have a negligible overall disutility for agent $2$ compared with $d_2(A_2)$.
On the other hand, there is only one single ``large item'' $j_{2}^{(1)}$ by our construction, as an item with disutilty $a_{n+1}$ can never appear if agent $2$ must take at least one item for every $n$ consecutive iterations.
This implies $d_2(j_2^{(1)})$ is also approximately the value of $\MMS_2$ by the time agent $2$'s disutility reaches $n\cdot d_2(j_2^{(1)})$.
We then have $d_2(A_2)\approx n\cdot\MMS_2$.

We have seen that we can construct the valuations for $d_1$ and $d_2$ such that one agent $i$'s disutility will approach to $n\cdot\MMS_i$ if the algorithm insists on only allocating items among the first two agents.
This implies that, for some finite integer $T(2)$, we have to allocate an item to an agent from $\{3,\ldots,n\}$ after $T(2)-1$ items have been allocated among agents $1$ and $2$.
If we consider that the algorithm only allocates items among the first three agents $\{1,2,3\}$, we know that agent $3$ must take at least one item in the first $T(2)$ iterations.
We can further construct $d_1$ and $d_2$ such that agent $3$ must take at least one item in \emph{every} $T(2)$ iterations.
To achieve this, whenever agent $3$ takes an item, we can set the disutilities of the next items to agents $1$ and $2$ to be a large number $W$ such that all previous items have negligible disutilities.
Then we can replicate the construction in the proof above, but with every disutility multiplied by $W$, until agent $3$ takes an item again.

Having established that agent $3$ must take at least one item in every $T(2)$ consecutive iterations, we construct $d_3$ such that we will eventually have $d_3(A_3)\approx n\cdot\MMS_3$.
This can be done by a construction similar to $d_2$ in the proof above, with the $a$-sequence satisfying $a_{T(2)+1}=d_3(j_3^{(1)})$.
This shows that, for some finite integer $T(3)$, at least one of the first three agents $i\in\{1,2,3\}$ will have $d_i(A_i)\approx n\cdot\MMS_i$ after $T(3)$ iterations.
Therefore, we must let an agent outside $\{1,2,3\}$ take an item after at most $T(3)-1$ iterations.
This gives a natural inductive proof for Theorem~\ref{thm:lowerbound}.

%% file: 51-negative.tex
\subsection{Proof of Theorem~\ref{thm:lowerbound}, General Case}
\label{sect:negative-general}
We prove Theorem~\ref{thm:lowerbound} in this section by formalizing the ideas described at the end of the previous section.
For each $n'=1,\ldots,n$ and each $\epsilon\in(0,1)$, define $T(n',\epsilon)$ be the integer such that
\begin{itemize}
    \item for any online algorithm that only allocates items among agents $\{1,\ldots,n'\}$, there exists a construction of $d_1,\ldots,d_{n'}$ (that can be adaptive on the algorithm's decision in items allocation for previous iterations) such that we have $d_i(A_i)>(n-\epsilon)\cdot\MMS_i$ for some $i\in[n']$ after some iteration $j$ with $j\leq T(n',\epsilon)$.
\end{itemize}
Our objective is to show that $T(n,\epsilon)$ is finite for any $\epsilon\in(0,1)$.
We will prove that $T(n',\epsilon)$ is finite for any $n'$ and any $\epsilon$ by induction on $n'$.
We will always set the disutility of the first item to $1$ for every agent.

The base step for $n'=1$ is trivial: we have $T(1,\epsilon)=n$ for any $\epsilon$, as we can construct $d_1$ such that all items have the same disutility $1$.

For the inductive step, suppose $T(n',\epsilon)$ is finite for any $\epsilon\in(0,1)$; we will show that $T(n'+1,\epsilon)$ is finite for any $\epsilon\in(0,1)$.
Fix an arbitrary online algorithm that only allocates items among agents in $\{1,\ldots,n',n'+1\}$, and the discussions in the remaining part of this proof are with respect to this algorithm.
We first define the disutility functions for the first $n'$ agents, and assume for this moment that $d_{n'+1}$ is arbitrary.
Let $\epsilon_{1\sim n'}=\frac\epsilon{n}$.
Let $d_1',\ldots,d_{n'}'$ be the construction of the disutility functions corresponding to the induction hypothesis on $T(n',\epsilon_{1\sim n'})$.
For each $i=1,\ldots,n'$, we define $d_i$ by defining $d_i(j)$ for each item $j$ as follows.
Let
$$d_i(j)=d_i'(j)$$
if agent $n'+1$ has not taken any item before the $j$-th iteration;
otherwise, let
$$d_i(j)=d_i'(j-j^\ast)\cdot \frac{d_i([j^\ast])}{\epsilon_{1\sim n'}},$$
where $j^\ast$ is the most recent iteration that agent $n'+1$ has taken an item.
Notice that the above equation is well-defined: the algorithm only allocates items among $\{1,\ldots,n'\}$ after the $j^\ast$-th iteration, and we can view the allocation of items in $[j]\setminus[j^\ast]$ as an output of a ``sub-algorithm'' that only allocates items among $\{1,\ldots,n'\}$, so $d_i'(j-j^\ast)$ in the equation above is well-defined.

We first show the following proposition, which says that agent $n'+1$ must take at least one item in every $T(n',\epsilon_{1\sim n'})$ consecutive iterations.

\begin{proposition}
    For $d_1,\ldots,d_{n'}$ defined above, regardless of the definition of $d_{n'+1}$, if there exists $T(n',\epsilon_{1\sim n'})$ consecutive iterations where the algorithm only allocates items among agents in $\{1,\ldots,n'\}$, there exists an agent $i\in[n']$ such that $d_i(A_i)> (n-\epsilon)\cdot\MMS_i$ in one of these $T(n',\epsilon_{1\sim n'})$ iterations. 
\end{proposition}
\begin{proof}
    If these $T(n',\epsilon_{1\sim n'})$ iterations are the first $T(n',\epsilon_{1\sim n'})$ iterations, the proposition can be easily implied by the induction hypothesis, by noticing $\epsilon_{1\sim n'}<\epsilon$.
    Otherwise, let $j^\ast$ be the iteration where agent $n'+1$ has taken an item and the next $T(n',\epsilon_{1\sim n'})$ items are allocated among $\{1,\ldots,n'\}$.
    For every $j>j^\ast$, let $A_i'=A_i\cap\left([j]\setminus [j^\ast]\right)$ be the set of items allocated to agent $i$ after the $j^\ast$-th iteration.
    Let $\MMS_i'$ be the $\MMS$ threshold for agent $i$ defined based on item set $[j]\setminus[j^\ast]$ and disutility function $d_i'$.
    By the induction hypothesis, we have $d_i'(A_i')>(n-\epsilon_{1\sim n'})\MMS_i'$ for some agent $i$ and at some iteration $j>j^\ast$.
    The discussions below are regarding this agent $i$ and this iteration $j$.
    
    By our definition of $d_i$, we have 
    $$d_i(A_i)\geq d_i\left(A_i'\right)=\frac{d_i([j^\ast])}{\epsilon_{1\sim n'}}\cdot d_i'(A_i')>\frac{d_i([j^\ast])}{\epsilon_{1\sim n'}}\cdot(n-\epsilon_{1\sim n'})\MMS_i'.$$
    On the other hand, we have $\MMS_i\leq d_i([j^\ast])+\frac{d_i([j^\ast])}{\epsilon_{1\sim n'}}\cdot\MMS_i'$ by considering adding the item set $[j^\ast]$ to the bundle with the largest disutility in the $n$-partition of $[j]\setminus[j^\ast]$ defining $\MMS_i'$.
    Moreover, since $d_i'(1)=1$, we have $\MMS_i'\geq1$, so $\MMS_i\leq d_i([j^\ast])\cdot(1+\frac1{\epsilon_{1\sim n'}}\MMS_i')\leq d_i([j^\ast])\cdot(1+\frac1{\epsilon_{1\sim n'}})\cdot \MMS_i'$.
    Substituting this to $d_i(A_i)$, we have
    $$d_i(A_i)>\frac{d_i([j^\ast])}{\epsilon_{1\sim n'}}(n-\epsilon_{1\sim n'})\cdot\frac{1}{d_i([j^\ast])\cdot (1+\frac1{\epsilon_{1\sim n'}})}\MMS_i=\frac{n-\epsilon_{1\sim n'}}{1+\epsilon_{1\sim n'}}\cdot\MMS_i>(n-\epsilon)\MMS_i,$$
    where the last inequality is easily implied by $\epsilon_{1\sim n'}=\frac\epsilon{n}$.
\end{proof}

To conclude the inductive step, it remains to construct the disutility function $d_{n'+1}$ such that we have $d_{n'+1}(A_{n'+1})>(n-\epsilon)\cdot\MMS_{n'+1}$ after finitely many iterations, if it is required that agent $n'+1$ takes at least one item in every $T(n',\epsilon_{1\sim n'})$ consecutive iterations.
Notice that this requirement implies $j_{n'+1}^{(1)}$ must exist (in fact, $j_{n'+1}^{(1)}\leq T(n',\epsilon_{1\sim n'})$).
Set $\epsilon_{n'+1}=\frac\epsilon{n(n+3)}$.

We first define $d_{n'+1}(j)$ for every $j\leq j_{n'+1}^{(1)}$.
Let
$d_{n'+1}(1)=1$ and $d_{n'+1}(j)=d_{n'+1}([j-1])/\epsilon_{n'+1}$ for every $j=2,\ldots,j_{n'+1}^{(1)}$.
We then define $d_{n'+1}(j)$ for each $j>j_{n'+1}^{(1)}$.
Let $\{a_1,a_2,a_3,\ldots\}$ be an increasing sequence of positive rational numbers such that
\begin{enumerate}
    \item $a_t>\frac1{\epsilon_{n'+1}}\sum_{t'=1}^{t-1}a_{t'}$ holds for every $t\geq 2$, and
    \item $a_{T(n',\epsilon_{1\sim n'})+1}=d_i(j_{n'+1}^{(1)})$.
\end{enumerate}
For each $j> j_{n'+1}^{(1)}$, let $d_{n'+1}(j)=a_s$ if the most recent iteration where agent $n'+1$ takes an item is at iteration $j-s$.
This completes the definition of $d_{n'+1}$.

To complete the proof for the inductive step, we will show that $d_{n'+1}(A_{n'+1})>(n-\epsilon)\cdot\MMS_{n'+1}$ right after the iteration where $d_{n'+1}(A_{n'+1})\geq n\cdot d_{n'+1}(j_{n'+1}^{(1)})$.
Notice that this is sufficient: there will be an iteration where $d_{n'+1}(A_{n'+1})\geq n\cdot d_{n'+1}(j_{n'+1}^{(1)})$, as each item taken by agent $n'+1$ has disutility at least $a_1$ (which is a fixed positive number) and agent $n'+1$ needs to take at least one item in every $T(n',\epsilon_{1\sim n'})$ consecutive iterations (which means there is no iteration after which agent $n'+1$ stops taking items).
Given this lower bound to $d_{n'+1}(A_{n'+1})$, we will conclude the proof by finding an upper bound to $\MMS_{n'+1}$.
In the remaining part of the proof, we let $V=d_{n'+1}(j_{n'+1}^{(1)})$ for notation simplicity, and let $j^\dag$ be the first iteration after which $d_{n'+1}(A_{n'+1})\geq n\cdot V$.
Below, $A_{n'+1}$ and $\MMS_{n'+1}$ are all with respect to $j^\dag$.

We state the following two straightforward observations:
\begin{itemize}
    \item[O1:] $\epsilon_{n'+1}\cdot d_{n'+1}(A_{n'+1})\geq d_{n'+1}([j^\dag]\setminus A_{n'+1})$;
    \item[O2:] for any $j\neq j_{n'+1}^{(1)}$, we have $d_{n'+1}(j)\leq\epsilon_{n'+1}\cdot V$.
\end{itemize}
Intuitively, O1 says that items not in $A_{n'+1}$ have a negligible overall disutility for agent $n'+1$, and O2 says that every item except for $j_{n'+1}^{(1)}$ has a small disutility.
To see O1, notice that we have $d_{n'+1}([j_{n'+1}^{(1)}-1])\leq \epsilon_{n'+1}\cdot V$ by our definition of $d_{n'+1}$, and, for any $w>1$, by our definition of the $a$-sequence, the item set $[j_{n'+1}^{(w)}-1]\setminus [j_{n'+1}^{(w-1)}]$ has a disutility of at most $\epsilon_{n'+1}\cdot a_{s}$ for $s=j_{n'+1}^{(w)}-j_{n'+1}^{(w-1)}$ (i.e., the set of items allocated to the other agents between two iterations where agent $n'+1$ takes items has a small overall disutility that is at most $\epsilon_{n'+1}$ times agent $(n'+1)$'s disutility for the next item received).
Therefore, whenever there is a time period where agent $n'+1$ is not taking items, the next item received by agent $n'+1$ ``dominates'' the items that are received by the others within this time period.
This concludes O1.
For O2, its validity for $j<j_{n'+1}^{(1)}$ follows by definition of $d_{n'+1}$, and its validity for $j>j_{n'+1}^{(1)}$ is due to that agent $n'+1$ needs to take at least one item in every $T(n',\epsilon_{1\sim n'})$ consecutive iterations (so that every item received after $j_{n'+1}^{(1)}$ has value at most $a_{T(n',\epsilon_{1\sim n'})}$, which is less than $\epsilon_{n'+1}\cdot a_{T(n',\epsilon_{1\sim n'})+1}=\epsilon_{n'+1}\cdot V$).

Next, we prove the following upper bound $\MMS_{n'+1}<(1+(n+3)\epsilon_{n'+1})V$.
We first find an $n$-partition of $A_{n'+1}$ such that each set in the partition has disutility at most $(1+2\epsilon_{n'+1})V$, and we assume next that all items in $[j^\dag]\setminus A_{n'+1}$ are put in the set with the highest overall disutility.
This gives an $n$-partition of $[j^\dag]$ such that each set's overall disutility is bounded by $(1+2\epsilon_{n'+1})V+\epsilon_{n'+1}\cdot d_{n'+1}(A_{n'+1})$ due to O1.
By O2 and the definition of $j^\dag$, we have $d_{n'+1}(A_{n'+1})= d_{n'+1}([j^\dag-1]\cap A_{n'+1})+d_{n'+1}(j^\dag)< nV+\epsilon_{n'+1}V$.
Thus, this bound becomes
$(1+2\epsilon_{n'+1})V+\epsilon_{n'+1}\cdot(nV+\epsilon_{n'+1}V)<(1+(n+3)\epsilon_{n'+1})V$.
Therefore, it remains to show that there exists an $n$-partition of $A_{n'+1}$ such that each set in the partition has disutility at most $(1+2\epsilon_{n'+1})V$.
To see this, consider the greedy bin-packing algorithm that puts items into $n$ bins with size $(1+2\epsilon_{n'+1})V$ for each bin, and the greedy algorithm starts using the next bin only if no item can be put into any previous bins.
By O2, it is easy to see that every almost-full bin contains items of overall size at least $(1+\epsilon_{n'+1})V$.
If all bins are almost full, the total size of all items are at least $n(1+\epsilon_{n'+1})V>nV+\epsilon_{n'+1}V\geq d_{n'+1}(A_{n'+1})$, leading to a contradiction.
Therefore, the greedy bin-packing algorithm can terminate with at least one bin not being almost full.
This implies we can put all items into the $n$ bins.
We thus conclude that $\MMS_{n'+1}<(1+(n+3)\epsilon_{n'+1})V$.

Finally, since $d_{n'+1}(A_{n'+1})\geq nV$, we have
$$d_{n'+1}(A_{n'+1})-(n-\epsilon)\MMS_{n'+1}>nV-(n-\epsilon)(1+(n+3)\epsilon_{n'+1})V>0,$$
where the last inequality is due to $\epsilon_{n'+1}=\frac\epsilon{n(n+3)}$.
This concludes the inductive step.

%% file: 60-2value.tex
\section{Optimize the Ratio in the Personalized Bi-Value Case}
In this section, we present a specialized algorithm to handle an important special case: the \emph{personalized bi-value} setting, where each agent has at most two distinct disutility values. In this case, we carefully optimize the constant in the competitive ratio through two key strategies:
\begin{itemize}
    \item We apply a more refined rounding scheme, rather than simply rounding to the nearest power of~$2$. Specifically, the rounding depends on the gap between the two disutility values of each agent: if they are close, we treat them as identical; if they are far apart, we avoid rounding altogether.
    \item In the generalized version, we allow $a$ and $b$ to be arbitrary parameters to capture all possible values of $n$ and to enable a clean, general formulation. However, note that when we reduce the original problem to the \game, we always have $a = 1$ and $b = \frac{1}{n-1}$. To optimize the constants, we can perform a more careful calculation that leverages this choice of $a$ and $b$.
    %
\end{itemize}

Formally speaking, on the algorithm side, we only modify the value-rounding procedure while keeping the remaining steps unchanged. Instead of simply rounding each disutility to the nearest power of two, we decide between two operations when each agent encounters her second distinct disutility value:
\begin{itemize}
    \item If $V_i^1 / V_i^2 \in [1/\alpha, \alpha]$, where $\alpha = \frac{\sqrt{3} -1}{2}$, we treat the two types of items as the same type, and only use one pressure parameter for it. 
    \item Otherwise, we treat the two types separately and associate different pressure parameters for them. 
\end{itemize}
We first state a more general version of \Cref{lem:bound_F}, tailored to leverage the specific choice $a = 1$ and $b = \frac{1}{n-1}$ in the reduction. The proof is almost identical to the previous one, except that we use the stronger expression $a+b \leq 1 + \frac{1}{n-1}=\frac{n}{n-1}$ instead of the earlier bound $a+b \leq 2$. For completeness, we include the proof at the end of this section.

\begin{lemma} [Generalized Version of \Cref{lem:bound_F}]
    \label{lem:F2}
    If further assume $a+b \leq \beta$, with some $0 \leq \beta \leq 2$, 
    $\forall x\in (-\frac12,\frac12]$, $t\in Z^+$, $F^t(x)\leq \frac{\beta k}{4}-\beta kx^2$. 
\end{lemma}


\begin{theorem}
    \label{thm:bi-value}
    There exists an algorithm that guarantees a $\min\{n, 2+\sqrt{3}\}$-MMS allocation for the online fair division problem under indivisible chores, in the special case where each agent has at most $2$ distinct disutilities.
\end{theorem}
\begin{proof}
    According to Lemma~\ref{lem:F2}, by fixing $k=2$ and $\beta = \frac{n}{n-1}$, we have that, for all $t$ and $x$, 
    $$
        F^t(x)\leq \frac{n}{2(n-1)}-\frac{2n}{n-1}x^2~.
    $$
    According to the reduction from the algorithm to the \game, for all $i\in [n]$ and $\u\in \{1,2\}$, 
    \begin{align*}
        H_i^\u &= \frac{F^t(\frac{1}{2}-\frac{1}{2n})}{\frac{1}{2n}}\\
        &= 2n\cdot [\frac{n}{2(n-1)}-\frac{2n}{n-1}(\frac{1}{4}-\frac{1}{2n}+\frac{1}{4n^2})]\tag{$F^t(x)\leq \frac{n}{2(n-1)}-\frac{2n}{n-1}x^2$}\\
        &= \frac{2n}{n-1}-\frac{1}{n-1}\\
        &= 2+\frac{1}{n-1}
    \end{align*}
    Then we analyze the number of type-$\u$ items $i$ takes in the output allocation. 
    According to the definition of $H_i^\u$, 
    \begin{align*}
    H_i^\u 
    &= \bigl|\text{items in } M_i^\u \text{ allocated to } i\bigr| 
       - \frac{1}{n-1} \cdot \bigl|\text{items in } M_i^\u \text{ allocated to other agents}\bigr| \\
    &= |A_i \cap M_i^\u| 
       - \frac{1}{n-1} \cdot \bigl(N_i^\u - |A_i \cap M_i^\u|\bigr) 
    = \frac{n}{n-1} |A_i \cap M_i^\u| - \frac{1}{n-1} N_i^\u.
    \end{align*}
    The number of type-$\u$ items agent $i$ gets is 
    $$|A_i\cap M_i^\u| \le \frac{N_i^\u}{n}+\frac{n-1}{n}H_i^\u\le
    \frac{N_i^\u}{n} + \frac{n-1}{n}\cdot (2+\frac{1}{n-1}) = \frac{N_i^\u}{n} + 2-\frac1n<\lceil N_i^\u/n\rceil + 2~.$$
    Since $|A_i\cap M_i^j|$ is an integer, we can obtain that
    $$
    |A_i\cap M_i^\u| \le \lceil N_i^\u/n\rceil +1~.
    $$ 
    Thus, 
    $\forall \u\in \{1,2\}$, 
    $$
    V_i^\u \cdot |A_i\cap M_i^\u| \leq (\lceil N_i^\u/n\rceil +1)\cdot V_i^\u=\MMS_i^\u+V_i^\u~.
    $$

Then we discuss the cases according to $V_i^1,V_i^2$ are of huge gap or not as follows. Without loss of generality, assume that $V_i^2$ is the smaller one and $V_i^2=\alpha\cdot V_i^1$ where $\alpha\in [0,1]$.
\begin{itemize}
    \item \textbf{Case I: $\alpha\leq \frac{\sqrt{3}-1}{2}$.}
    We have 
    \begin{align*}
        d(A_i) &= V_i^1 \cdot |A_i\cap M_i^1| + V_i^2 \cdot |A_i\cap M_i^2|\\
        &\leq (\MMS_i^1+V_i^1)+(\MMS_i^2+V_i^2)\\
        &= \MMS_i^1+\MMS_i^2+(1+\alpha)\cdot V_i^1\\
        &\leq \MMS_i+2(1+\alpha)\cdot V_i^1\tag{\Cref{lem:mms_decomp}}\\
        &\leq (3+2\alpha)\MMS_i\tag{$\MMS_i\geq\MMS_i^1\geq V_i^1$}\\
        &\leq (2+\sqrt{3})\MMS_i ~.\tag{$\alpha\leq \frac{\sqrt{3}-1}{2}$}
    \end{align*}
    \item \textbf{Case II: $\alpha> \frac{\sqrt{3}-1}{2}$.}
    In this case, the items with disutilities $V_i^1$ and $V_i^2$ are treated as a single type. Without loss of generality, we imagine that the algorithm treats their disutility as $V_i^1$, and we use $\MMS_i'$ to denote the MMS benchmark under the assumption that all such items indeed have disutility $V_i^1$. Then we have
    \begin{align*}
        d(A_i) 
        &\leq \MMS_i' + V_i^1 
        \leq \frac{\MMS_i}{\alpha} + V_i^1 
        \leq (2 + \sqrt{3}) \MMS_i~.
    \end{align*}

\end{itemize}
The discussion of the two cases concludes the lemma. 
\end{proof}

\begin{proof}[Proof of \Cref{lem:F2}]
Using the same argument as in \Cref{lem:contiguous_interval}, it suffices to consider the contiguous operation.
Hence, it remains to show the inductive step -- the upper bound still holds after every contiguous operation.
Our proof mainly follows the proof of \Cref{lem:ub_of_contiguous_interval}.
Similarly, using the argument on symmetry, it suffices to consider the case that $x$ is located to the left of $B$.
We define the upgrading and downgrading points in the same way as before.
Thereafter, we apply the same scaling method by multiplying \eqref{eqn:f_t_1} by $\delta+\ell$ and \eqref{ineq:ub_of_F_t_1} by $\ell$, then summing, and obtain the same inequality.
\begin{align*}
(\delta + 2\ell)\, F^{t+1}(x) 
&\le (\delta + \ell)\, F^t(x-\ell) 
   + \ell\, F^t(x+\ell+\delta)
   + (\delta+2\ell)\, \Gamma~,
\end{align*}
where $$\Gamma = a \cdot \min(\ell + \delta, |A|) 
- b \cdot \max(|B| - \ell, 0)$$ follows the same calculation of Lemma~\ref{lem:ub_of_contiguous_interval}.
Using the new inductive hypothesis that $F^t(x) \le \beta\cdot k/4 - \beta kx^2$ for any $x\in I$.
Applying the inductive hypothesis, we can obtain that
\begin{align*}
(\delta+ 2\ell)F^{t+1}(x) \le (\delta+2\ell) \left(\frac{k}{4}\cdot \beta - \beta\cdot kx^2 - \beta\cdot k\ell(\ell+\delta) + \Gamma\right)\,.
\end{align*}
Then it suffices to prove $\Gamma\le \beta\cdot k\ell\cdot (\ell+\delta)$, which is an generalization of the inequality of $\Gamma \le 2k\ell\cdot (\ell+\delta)$.
In particular, the previous proof regards the upper bounds of $a$ and $b$ as one.
As a result, we use the inequality $a+b\leq 2$ in the later proof. Now, we can replace it with the stronger one $a + b \leq \beta$. 
The analysis of the four cases can still be applied under the stronger inequality. 

\begin{itemize}
    \item $\ell + \delta \leq \abs{A}$ and $\abs{B} - \ell \geq 0$:
    In this case, follows the same calculation of Lemma~\ref{lem:ub_of_contiguous_interval}, we have 
    \begin{align*}
        2k \ell \cdot (\ell+\delta)-\Gamma
        = (2k\ell-a)\cdot (\ell + \delta)+b\cdot(\abs{B}-\ell)~. 
    \end{align*}
    $2k \ell \cdot (\ell+\delta)-\Gamma \geq 0$ holds if
    $2k\ell-a\geq 0$. If $2k\ell-a<0$,  
    Define $r=\abs{A}-(\ell+\delta)\geq 0$,
    we have
    Following the calculation of \Cref{lem:ub_of_contiguous_interval}, we can obtain that 
    \begin{align*}
    \beta \cdot k\ell \cdot (\delta+\ell)-\Gamma
    &\geq  \beta \cdot k\ell \cdot \left(\frac{b}{k\,(a+b)}-r\right)-a\cdot (\abs{A}-r)+a\cdot\abs{A}-b\ell
\end{align*} 
Since $a+b\le \beta$, it follows that
\begin{align*}
\beta \cdot k\ell \cdot (\delta+\ell)-\Gamma &\geq  \beta k\ell \cdot \left(\frac{b}{k\cdot \beta}-r\right)-a\cdot (\abs{A}-r)+a\cdot\abs{A}-b\ell\tag{$a+b \leq \beta$}\\
    &= (a-\beta k\ell)\cdot r\geq 0~. \tag{$\beta k\ell-a<0$, $r\geq 0$}
\end{align*}

This concludes the case.
    \item $\ell+\delta\leq \abs{A}$ and $\abs{B} - \ell < 0$: In this case, since $\ell>\abs{B}=\frac{a}{k(a+b)}\geq \frac{a}{\beta \cdot k}$ by $a + b \leq \beta$, 
    $$
        \Gamma = a\cdot (\ell+\delta) \leq \beta k\ell \cdot (\ell+\delta)~.
    $$
    
    \item $\ell+\delta> \abs{A}$ and $\abs{B} - \ell \geq 0$: In this case, since $\ell+\delta>\abs{A}=\frac{b}{k(a+b)}\geq\frac{b}{\beta\cdot k}$ by $a+b \leq \beta$, 
    $$
         \Gamma = a\abs{A}-b\abs{B}+b\ell = b\ell \leq \beta k\ell\cdot(\ell+\delta)~.
    $$
    \item $\ell + \delta > |A|$ and $|B| - \ell < 0$: Using the same argument, the case is still impossible.
    \qedhere
\end{itemize}
\end{proof}